\documentclass{article}
\usepackage{bm}
\usepackage{cite}
\usepackage{mathrsfs}
\usepackage{amssymb}
\usepackage{epsfig,amsfonts,amsthm,amssymb,latexsym,amsmath}
\usepackage[ruled,vlined]{algorithm2e}
\usepackage[all]{xy}

\newtheorem{theorem}{Theorem}[section]

\newtheorem{example}[theorem]{Example}
\newtheorem{proposition}[theorem]{Proposition}

\newcommand{\cc}{\bm{c}}

\newcommand{\xx}{\bm{x}}
\newcommand{\yy}{\bm{y}}

\begin{document}

\title{Quasi-perfect codes in the $\ell_p$ metric \footnote{This work was partially supported by FAPESP, grants 2014/20602-8, 2013/25977-7 and CNPq, grant 312926/2013-8.}}
 
\author{Jo\~{a}o E. Strapasson, Grasiele C. Jorge, \\ Antonio Campello and Sueli I. R. Costa \footnote{\textit{address:} [Strapasson] UNICAMP - University of Campinas, 13484-350,  Limeira, SP, Brazil, [Jorge] UNIFESP - Federal University of S\~{a}o Paulo, 12231-280, S\~{a}o Jos\'{e} dos Campos, SP, Brazil and [Campello-Costa] UNICAMP - University of Campinas, 13083-859, Campinas, SP, Brazil} \footnote{\textit{email:} joao.strapasson@fca.unicamp.br, grasiele.jorge@unifesp.br, campello@ime.unicamp.br and sueli@ime.unicamp.br}}

\maketitle

----------------------------------------------------------------------

\begin{abstract} We consider quasi-perfect codes in $\mathbb{Z}^n$ over the $\ell_p$ metric, $2 \leq p < \infty$. Through a computational approach, we determine all radii for which there are linear quasi-perfect codes for $p = 2$ and $n = 2, 3$. Moreover, we study codes with a certain \textit{degree of imperfection}, a notion that generalizes the quasi-perfect codes. Numerical results concerning the codes with the smallest degree of imperfection are presented.   \end{abstract}

----------------------------------------------------------------------

{\bf keywords:}{Tilings, Lattices, Quasi-perfect Codes, $\ell_p$ metric}

\section{Introduction}

\hspace{0.5cm} A collection of disjoint translates of a set  $\mathcal{S} \subseteq \mathbb{Z}^{n}$ is called a \textit{tiling} of $\mathbb{Z}^n$ if the union of its elements is equal to $\mathbb{Z}^n$. We consider here $S = B_p^{n}(r)$ as the ball in $\mathbb{Z}^{n}$ with radius $r>0$ in the $\ell_p$ metric for $1 \leq p \leq \infty$. The set  $\mathscr{C} \subseteq \mathbb{Z}^{n}$ associated to the translations of such a tiling is also called a \textit{perfect code} in the $\ell_p$ metric, $1 \leq p \leq \infty$. If, in addition, this set is an additive subgroup $\Lambda$ of $\mathbb{Z}^n$, we call the corresponding tiling a \textit{lattice tiling}, and the corresponding code a \textit{linear perfect code}. 

For $p = 1$, the existence of  tilings by balls in the $\ell_1$ metric was investigated by Golomb and Welch in their seminal paper \cite{Golomb}. The so-called Golomb-Welch conjecture states that there are no tilings of $\mathbb{Z}^{n}$ by $B_1^{n}(r)$ for $n \geq 3$ and radius $r \geq 2$. Although there have been many advances and partial results towards a proof of this conjecture it still remains open (see \cite{Horak} for further references).


The existence of perfect codes  $\mathscr{C} \subseteq \mathbb{Z}^{n}$ in the $\ell_p$ metric, $2 \leq p \leq \infty$, with parameters $(n,r,p)$, where $n$ is the dimension and $r$ is the packing radius, was investigated in \cite{CNMAC2014,PerfectpLee}. It was shown that for $n=2,3$ and $p=2$ there are linear perfect codes only for the parameters $(2,r,2)$ and $r = 1, \sqrt{2}, 2, 2\sqrt{2}$ and $(3, r, 2)$ and $r = 1, \sqrt{3}$ \cite[Theorem 6.2 and 6.4]{PerfectpLee}. It was also shown that for $n=2$ and $r$ integer there are no perfect (linear and nonlinear) codes in the $\ell_p$ metric if $r>2$ and $2  \leq p < \infty$ \cite[Theorem 7.2]{PerfectpLee}. 

In view of the rarity of perfect codes in the $\ell_p$ metric, $2\leq p < \infty$, we relax the condition of being perfect by considering  quasi-perfect codes in the $\ell_p$ metric for $2 \leq p < \infty$ and by introducing the notion of \textit{degree of imperfection} of a code. Quasi-perfect codes in the Lee metric ($p=1$) have already been investigated in some papers. In \cite{AIBdaiwi-Bose} it was presented
quasi-perfect codes in the Lee metric for dimension $n = 2$. In 
\cite{Horak} the authors presented quasi-perfect codes
for $n = 3$ and a few radii. Later, in \cite{Camarero-Martinez} the authors constructed a family of quasi-perfect codes in the Lee metric of radius $2$ and arbitrarily large dimension. 

When dealing with quasi-perfect codes a natural question to be considered is the existane such codes for different $n$ and $p$, $2 \leq p < \infty$, for a given radius. In this paper we give partial answers for this question in some dimensions. Some preliminary results of this work were presented in \cite{CNMAC2014}. 

\subsection{Organization}

The paper is organized as follows. In Section \ref{CodesandLattices} we establish our notation and present some preliminary results concerning codes and lattices. In Section \ref{polio} a polyomino associated to a ball in the $\ell_p$ metric is considered.  In Section \ref{Quasi-Perfect} the notion of degree of imperfection of a code in the $\ell_p$ metric is introduced. In Section \ref{ExplicitConstruction} some families of lattices whose degree of imperfection is greater than $1$ are presented. 
Finally, in Section \ref{Algorithm} an algorithm  that searches for quasi-perfect codes in dimensions $2$ and $3$ is presented and  all quasi-perfect codes in dimension $2$ and $3$ for $p=2$ are listed.

\subsection{Relation to codes over finite alphabets}

In the classical literature, codes are considered over finite alphabets, for instance as subgroups of $\mathbb{F}_q^n$, where $\mathbb{F}_q$ is a field with $q$ elements, or as subgroups of $\mathbb{Z}_q^n$, where $\mathbb{Z}_q$ is the ring of integers modulo $q$. 

Codes in $\mathbb{Z}_q^n$ can be ``lifted'' to codes in $\mathbb{Z}^n$ via the $q$-ary Construction A \cite{SloaneLivro}. If the alphabet size $q$ is large enough, then perfect linear codes in $\mathbb{Z}_q^n$ in the $p$-Lee metric \cite{PerfectpLee} induce perfect linear codes in $\mathbb{Z}^n$ in the $\ell_p$ metric, as shown in \cite[Corollary 3.5]{PerfectpLee}. Conversely, proofs of non-existence of perfect codes in $\mathbb{Z}^n$ automatically imply the non-existence of codes in $\mathbb{Z}_q^n$ under certain conditions (for more precise definitions see \cite[Corollary 3.5]{PerfectpLee}). This relation, that dates back to Golomb and Welch \cite{Golomb} for the $\ell_1$ metric, motivates the study of codes over the alphabet $\mathbb{Z}$. It also justifies the terminology ``perfect'' and ``quasi-perfect'' codes (which, nonetheless, follows the terminology of \cite{Horak, AIBdaiwi-Bose}, etc.).
 
\section{Codes in the $\ell_p$ metric}\label{CodesandLattices}

A linear code, to our purposes, is an additive subgroup of $\mathbb{Z}^n$ (or a \textit{lattice}). We consider here full rank lattices in $\mathbb{Z}^n$, that is, full rank additive subgroups of $\mathbb{Z}^n$. A lattice $\Lambda$ always has a \textit{generator matrix} $B$, i.e., a full rank matrix such that $\Lambda = \left\{ \bm{x} B : \bm{x} \in \mathbb{Z}^n \right\}$. The {\it determinant} of a lattice is defined as $\det \Lambda = |\det B|$ for any generator matrix.

Recall that the $\ell_p$ distance between two points $\bm{x},\bm{y} \in \mathbb{Z}^n$ is defined as
\begin{equation} \label{dp} d_{p}({\bm x},{\bm y}) :=
\left(\displaystyle\sum_{i=1}^n|x_i-y_i|^{p}\right)^{1/p} \mbox{ if } 1 \leq p < \infty \end{equation}
and $d_{\infty}({\bm x},{\bm y}) := \max\left\{|x_i-y_i|; \,\, i=1,\ldots,n\right\}.$

The \textit{minimum distance} $d_p(\mathscr{C})$ of a code $\mathscr{C}$ in $\mathbb{Z}^{n}$ is defined as $$d_{p}(\mathscr{C}) = \min_{{\bm{\xx}},{\bm{\yy}} \in \mathscr{C}} d_{p}({\bm{\xx}},{\bm{\yy}}).$$ The minimum distance of a lattice $\Lambda$ in the $\ell_p$ metric, $d_p(\Lambda)$, is the  shortest nonzero vector in the $\ell_p$ metric.

Two lattices $\Lambda_1$ and $\Lambda_2$ in $\mathbb{Z}^{n}$ are congruence in the $\ell_p$ metric if $\Lambda_1$ can be obtained from $\Lambda_2$ by  permutation of coordinates composed with  sign changes.

\section{Balls in the $\ell_p$ metric and associated polyominoes} \label{polio}

In what follows $B_p^n(\bm{x},r)$ will be used for the closed ball in $\mathbb{Z}^n$ centered at $\bm{x}=(x_1,\ldots,x_n)$ with radius $r$, i.e., $B_{p}^n(\bm{x},r)= \left\{ (z_1, \ldots, z_n) \in \mathbb{Z}^n:\right.$ \linebreak 
$\left.|z_1-x_1|^p + \ldots + |z_n-x_n|^p \leq r^p \right\}$ for $1 \leq p < \infty$ and $B_{\infty}^n(\bm{x},r)=$ \linebreak $\{ (z_1, \ldots, z_n) \in \mathbb{Z}^n : \max\{|z_1-x_1|,\ldots,|z_n-x_n|\} \leq r \}.$ When ${\bm x} = {\bm 0}$ we will denote $B_{p}^n(\bm{0},r) = B_{p}^n(r)$ for $1 \leq p < \infty$ and $B_{\infty}^n(\bm{0},r) = B_{\infty}^n(r).$

Let $\mu(n,p,r)$ be the cardinality of the set $B_p^{n}(r) \cap \mathbb{Z}^{n}$. There is no closed form for $\mu(n,p,r)$ when $p \neq 1$ and $p \neq \infty$.

By considering the union of unit cubes in $\mathbb{R}^n$ centered at the points of $B_p^n(r)$, $1\leq p \leq \infty$, a shape called a \textit{polyomino} is produced. A tiling of $\mathbb{Z}^n$ by translates of $B_p^n(r)$ corresponds to a tiling of $\mathbb{R}^n$ by the associated polyominoes. We use the notation
\begin{equation} T_p^{n}(r) := \bigcup_{\xx \in B_p^n(r)}{\left( \xx + \left[\frac{-1}{2},\frac{1}{2}\right]^n\right)}, \,\, 1 \leq p \leq \infty,
\label{eq:defPol}
\end{equation}
for the polyomino in the $\ell_p$ metric associated to $B_p^n(r)$. Some polyominoes are depicted in Figure \ref{fig:poliominos}.

\begin{figure}[h!]
	\begin{minipage}[b]{0.15\linewidth}
		\centering
		\includegraphics[scale=0.17]{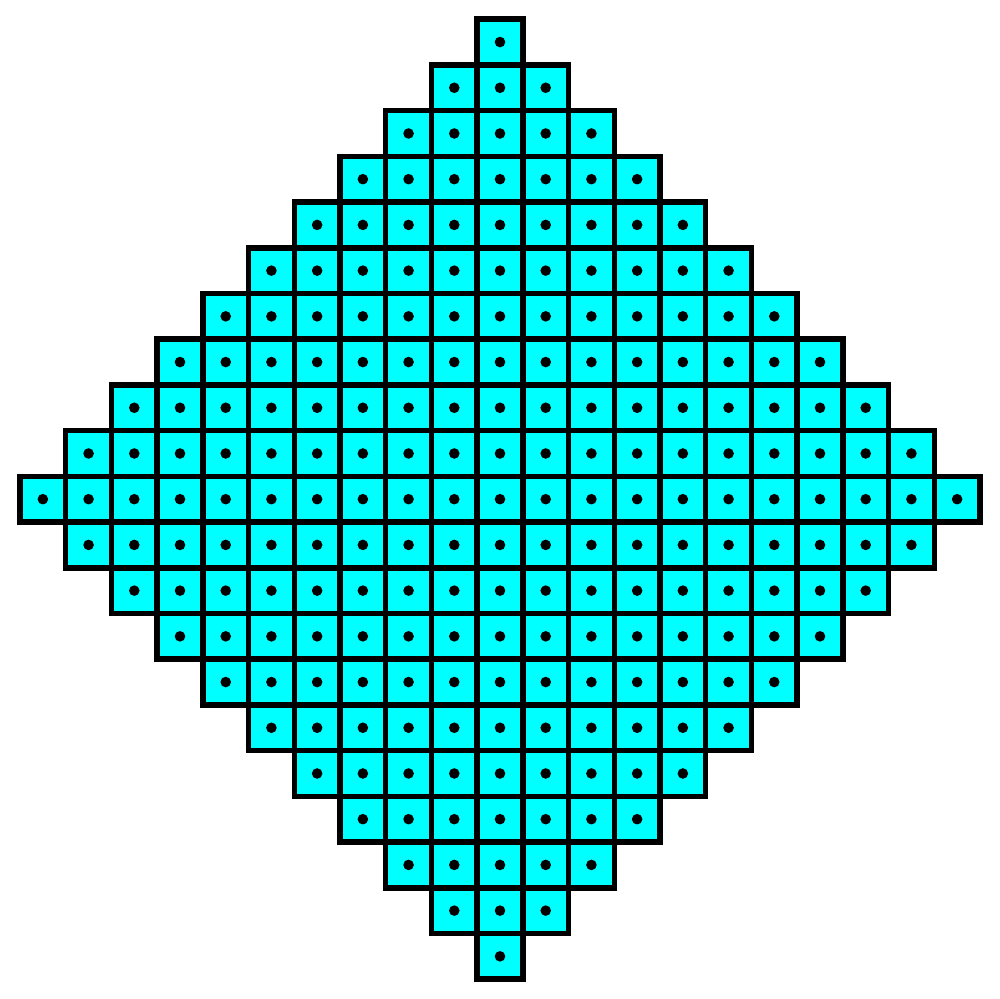}
	\end{minipage}
	\begin{minipage}[b]{0.15\linewidth}
		\centering
		\includegraphics[scale=0.17]{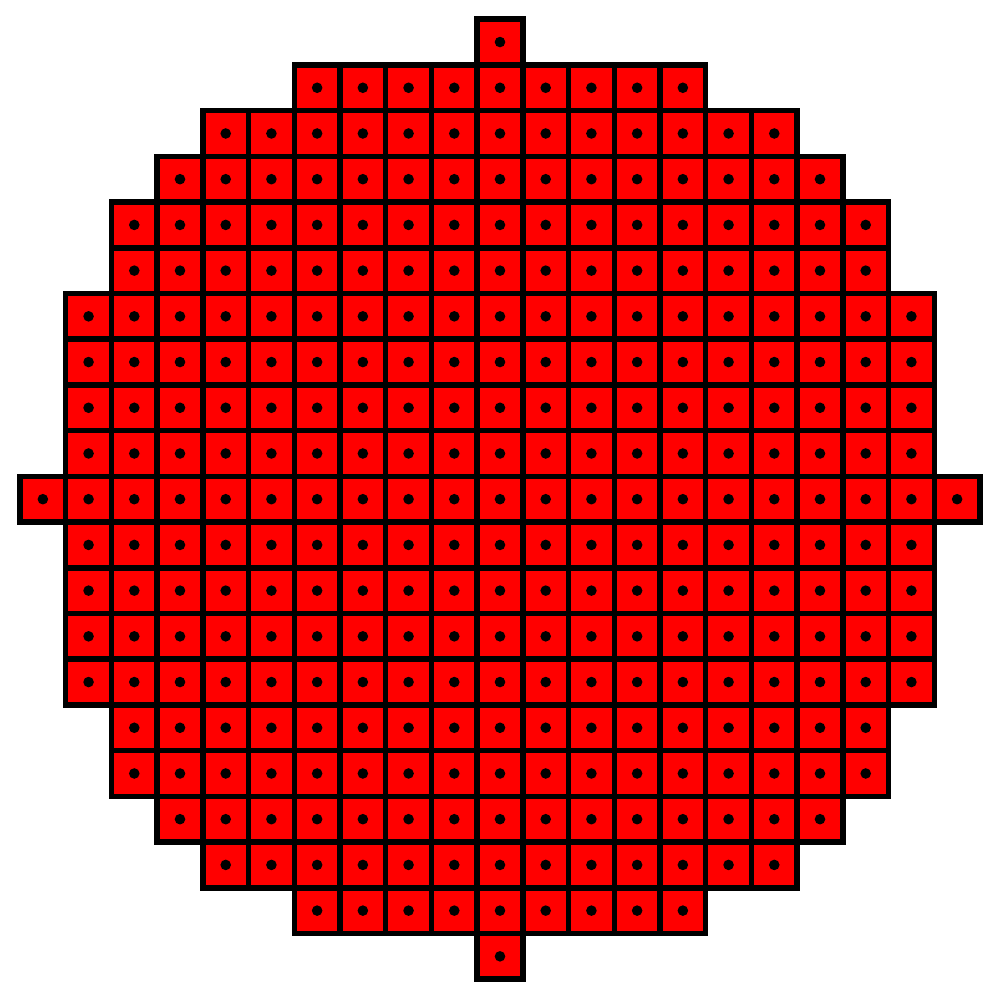}
	\end{minipage}
	\begin{minipage}[b]{0.15\linewidth}
		\centering
		\includegraphics[scale=0.17]{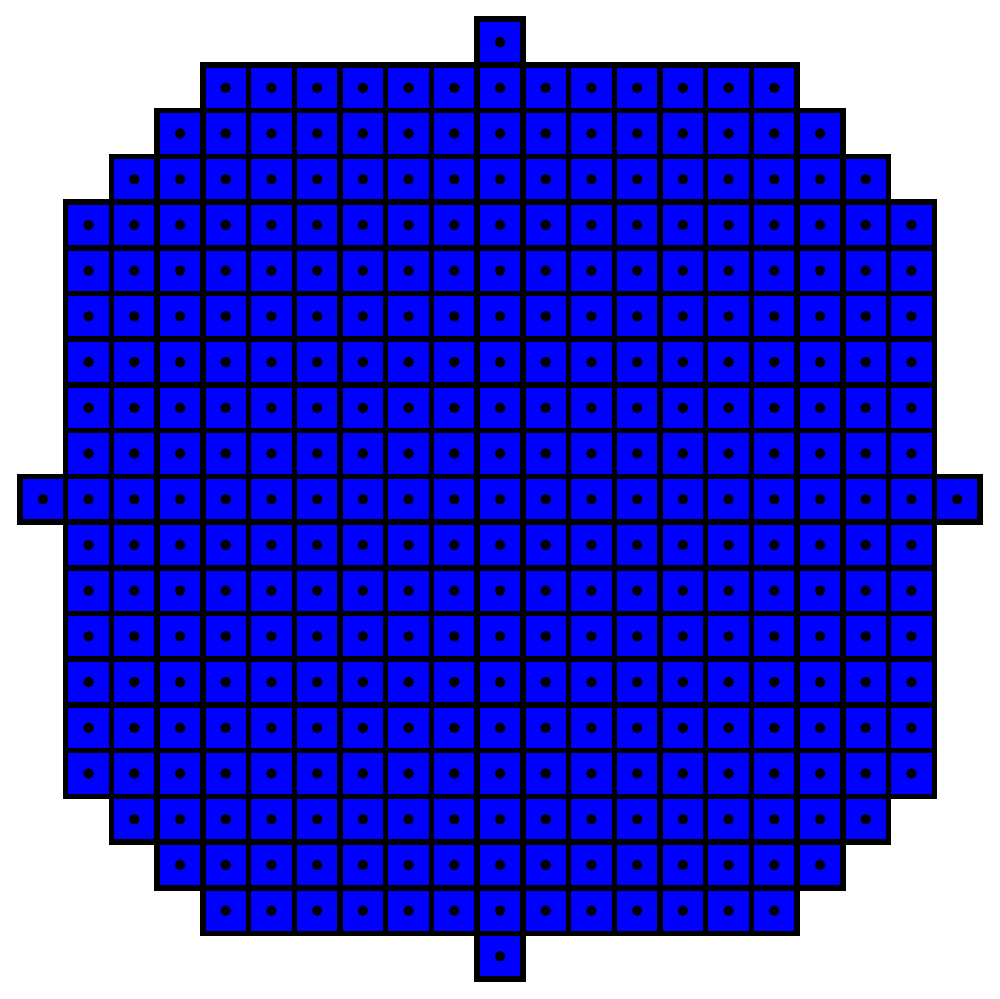}
	\end{minipage}
	\begin{minipage}[b]{0.15\linewidth}
		\centering
		\includegraphics[scale=0.17]{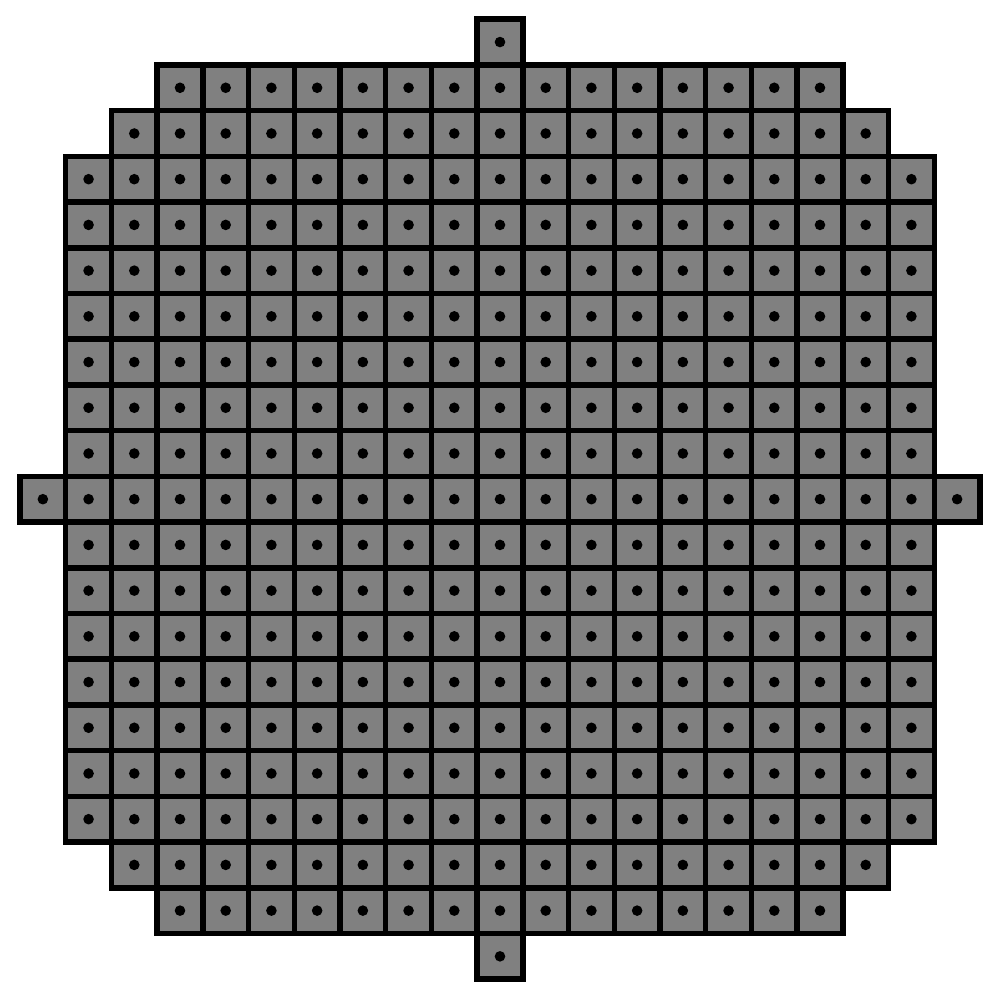}
	\end{minipage}
	\begin{minipage}[b]{0.15\linewidth}
		\centering
		\includegraphics[scale=0.17]{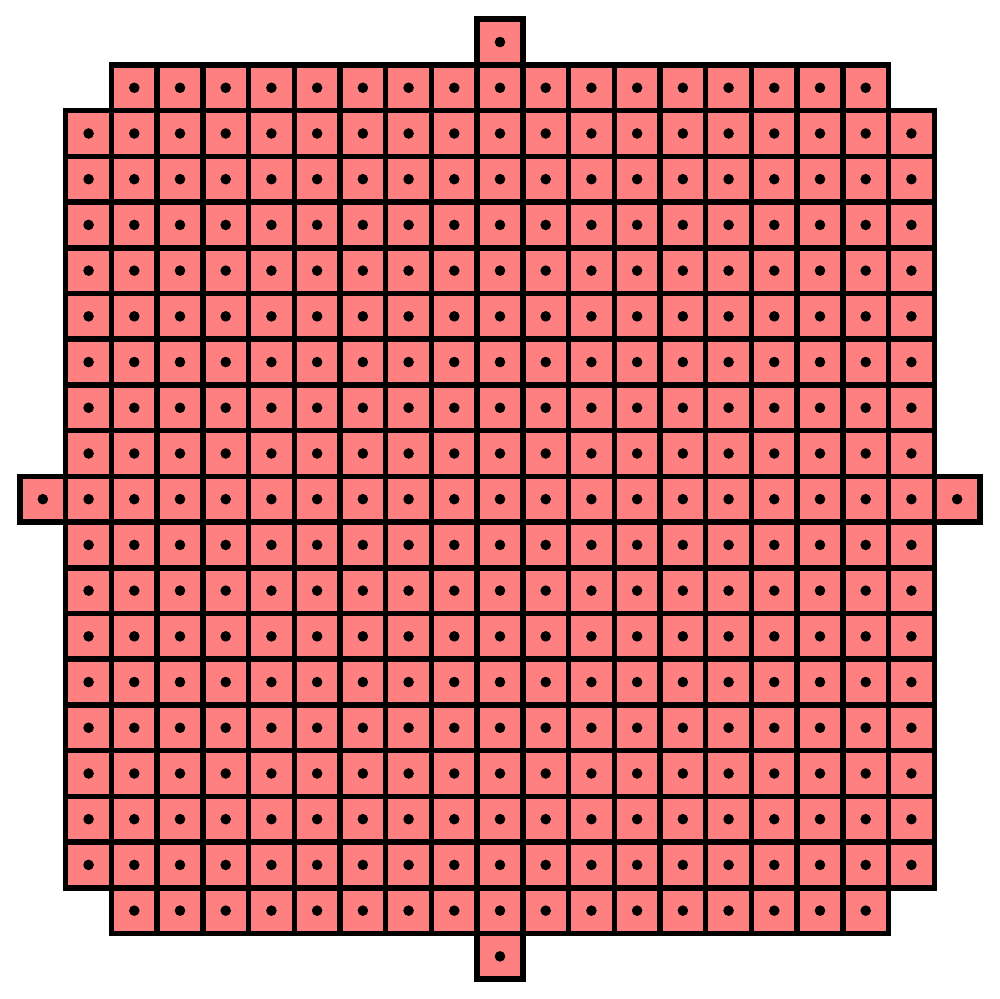}
	\end{minipage}
	\begin{minipage}[b]{0.15\linewidth}
		\centering
		\includegraphics[scale=0.17]{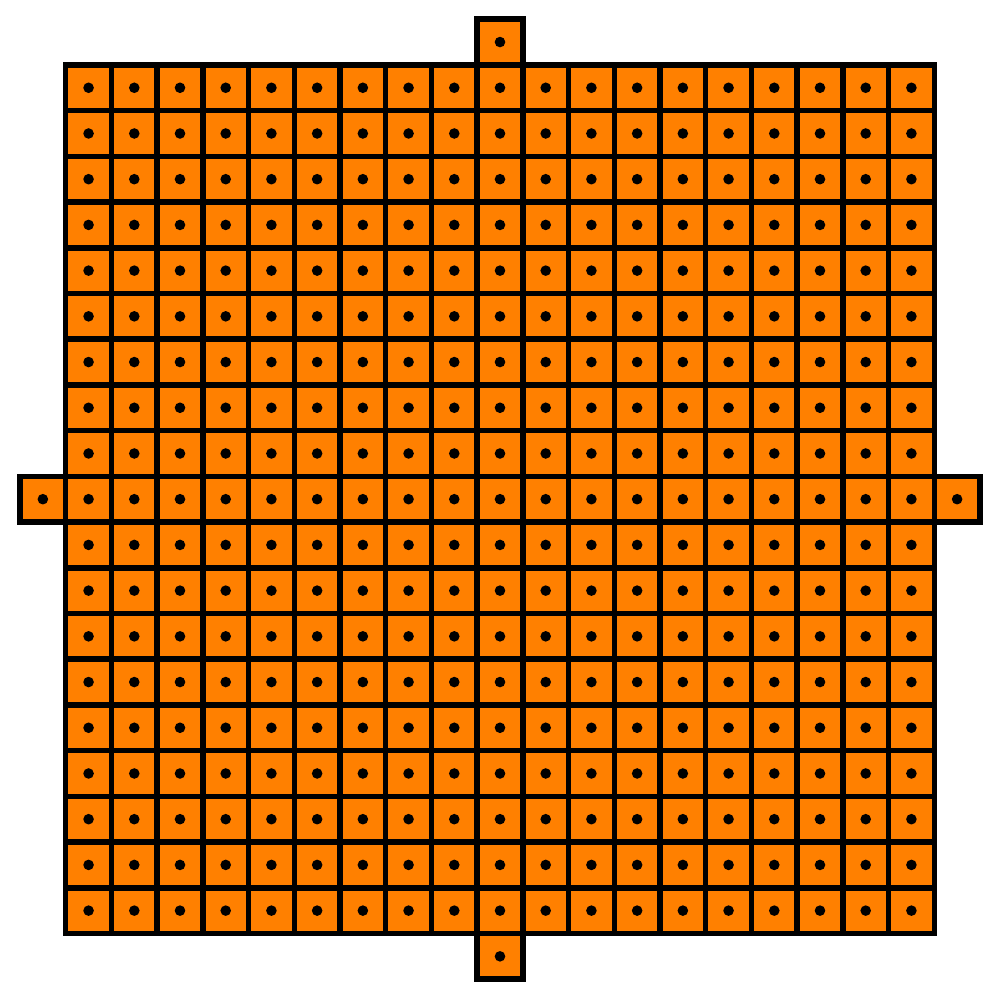}
	\end{minipage}
	\caption{From the left to the right: The polyominoes $T_1^{2}(10), T_2^{2}(10), T_3^{2}(10)$, $T_{4}^{2}(10)$, $T_{5}^{2}(10)$ and $T_7^{2}(10)$.}
	\label{fig:poliominos}
\end{figure}

When $n$, $p$ and $r$ vary, the polyomino $T_p^{n}(r)$ changes its shape. For $p=\infty$ we have squared polyominoes $T_{\infty}^{n}(r) = \left[r-\frac{1}{2},r+\frac{1}{2}\right]^{n}$. In Proposition \ref{pontos} we consider other special values for $n,r$ and $p$ (illustrated in Figure \ref{Fig2}) and obtain the shape of the polyominoes associated to $B_{p}^{n}(r)$.

From now on let ${\bm e_i}$ be the vector that has $1$ in the $i$-th entry and $0$ in the others, for $i=1,\ldots,n$.

\begin{proposition}\label{pontos} For fixed radius $r$ and  dimension $n$, we can assert:
	
	\begin{enumerate} \item[(i)] If $r$ is integer and $\frac{\ln n}{\ln
			\left(\frac{r}{r-1}\right)} \leq  p < \infty$, then $$T_{p}^{n}(r) = T_{\infty}^{n}(r-1) \bigcup{\left( \pm\bm{e_i} + \left[\frac{-1}{2},\frac{1}{2}\right]^n\right)}$$ and $\mu(n,p,r) = (2r-1)^{n} + 2n$. 
		\item[(ii)] If $r$ is integer, $p< \frac{\ln n}{\ln
			\left(\frac{r}{r-1}\right)}$ and  $(n-1)(r-1)^{p}+ (r-2) \leq r^{p}$,  then $$T_{p}^{n}(r) = \left(T_{\infty}^{n}(r-1) \bigcup{\left( \pm\bm{e}_i + \left[\frac{-1}{2},\frac{1}{2}\right]^n\right)}\right) \setminus  \bigcup_{\xx \in A} \left(\xx + \left[\frac{-1}{2},\frac{1}{2}\right]^n\right),$$ where $A =\{(\pm r, \ldots,\pm r) \}$ and $\mu(n,p,r) = (2r-1)^{n} + 2n - 2^{n}$. 
		\item[(iii)] If $r$ is not integer, $p< \frac{\ln n}{\ln
			\left(\frac{r}{\lfloor r \rfloor}\right)}$ and $(n-1)(\lfloor r \rfloor)^{p}+ (\lfloor r \rfloor -1)^{p} \leq r^{p}$, then $$T_p^{n}(r) = T_{\infty}^{n}(\lfloor r \rfloor) \setminus  \bigcup_{\xx \in B} \left(\xx + \left[\frac{-1}{2},\frac{1}{2}\right]^n\right),$$ where $B =\{(\pm \lfloor r \rfloor, \ldots,\pm \lfloor r \rfloor)\}$ and $\mu(n,p,r) = (2 \lfloor r \rfloor +1)^{n} - 2^n.$  
		\item[(iv)] If $r$ is not integer,  $(n-1)(\lfloor r \rfloor)^{p}+ (\lfloor r \rfloor -1)^{p} > r^{p}$ and $(n-1)(\lfloor r \rfloor)^{p}+ (\lfloor r \rfloor -2)^{p} \leq r^{p}$, then $$T_p^{n}(r) = T_{\infty}^{n}(\lfloor r \rfloor) \setminus  \bigcup_{\xx \in D} \left(\xx + \left[\frac{-1}{2},\frac{1}{2}\right]^n\right),$$ where $D =\{(\pm \lfloor r \rfloor, \ldots,\pm \lfloor r \rfloor)\} \cup \{(\pm \lfloor r-1 \rfloor, \ldots,\pm \lfloor r \rfloor)\} \cup \{(\pm \lfloor r \rfloor, \ldots,\pm \lfloor r-1 \rfloor)\}$ and $\mu(n,p,r) = (2 \lfloor r \rfloor +1)^{n} - (n+1)2^{n}.$  
	\end{enumerate} \end{proposition}
	
	\begin{proof} For the proof of (i) note that if $r$ is integer and $p\geq \frac{\ln n}{\ln \left(\frac{r}{r-1}\right)}$, then $n (r-1)^{p} \leq r^{p}$. Hence,   $B_{\infty}^{n}(r-1) \subseteq B_{p}^{n}(r) \subseteq B_{\infty}^{n}(r)$ since if $x \in B_{\infty}^{n}(r-1)$, then $\|x\|_p \leq n^{1/p}
		\|x\|_{\infty} \leq n^{1/p}(r-1) \leq r$. The proof of Items (ii), (iii) and (iv) is very similar. 
	\end{proof}

\begin{figure}[h!]
	\begin{minipage}[b]{0.15\linewidth}
		\centering
		\includegraphics[scale=0.17]{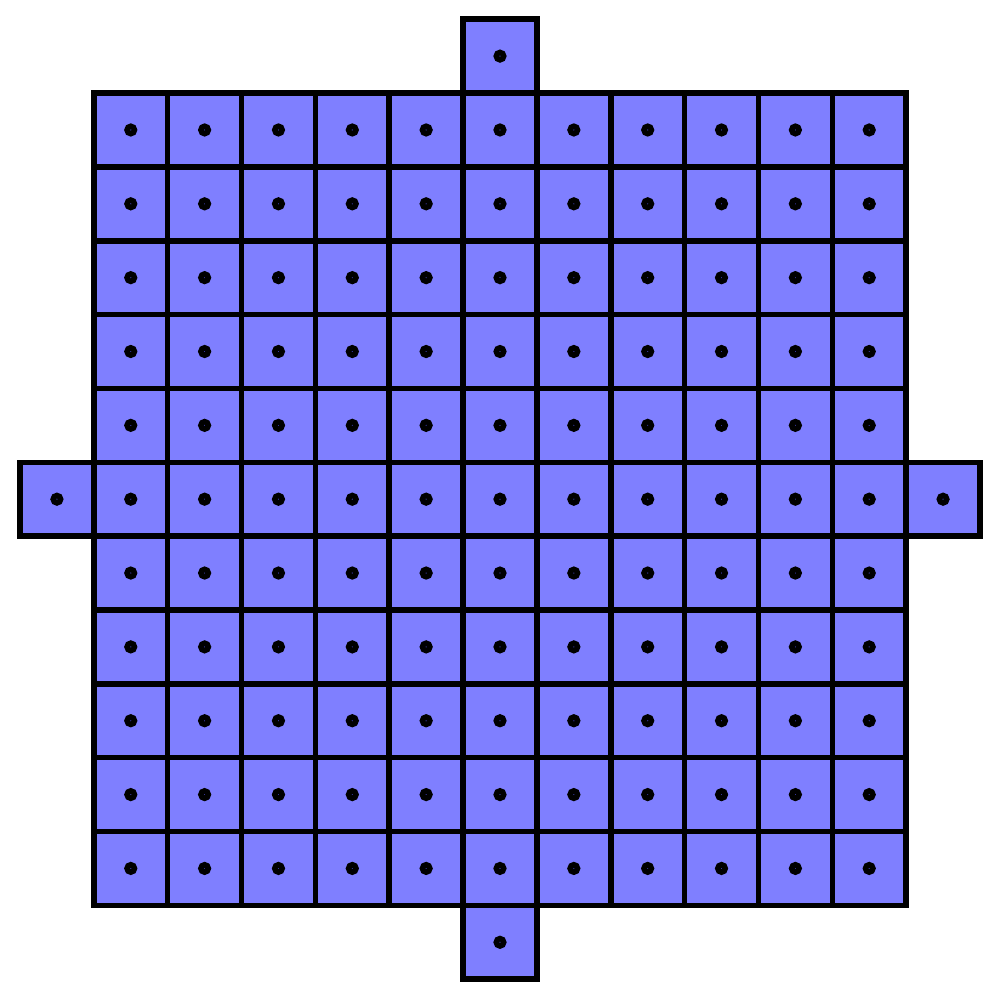}
	\end{minipage}
	\begin{minipage}[b]{0.15\linewidth}
		\centering
		\includegraphics[scale=0.17]{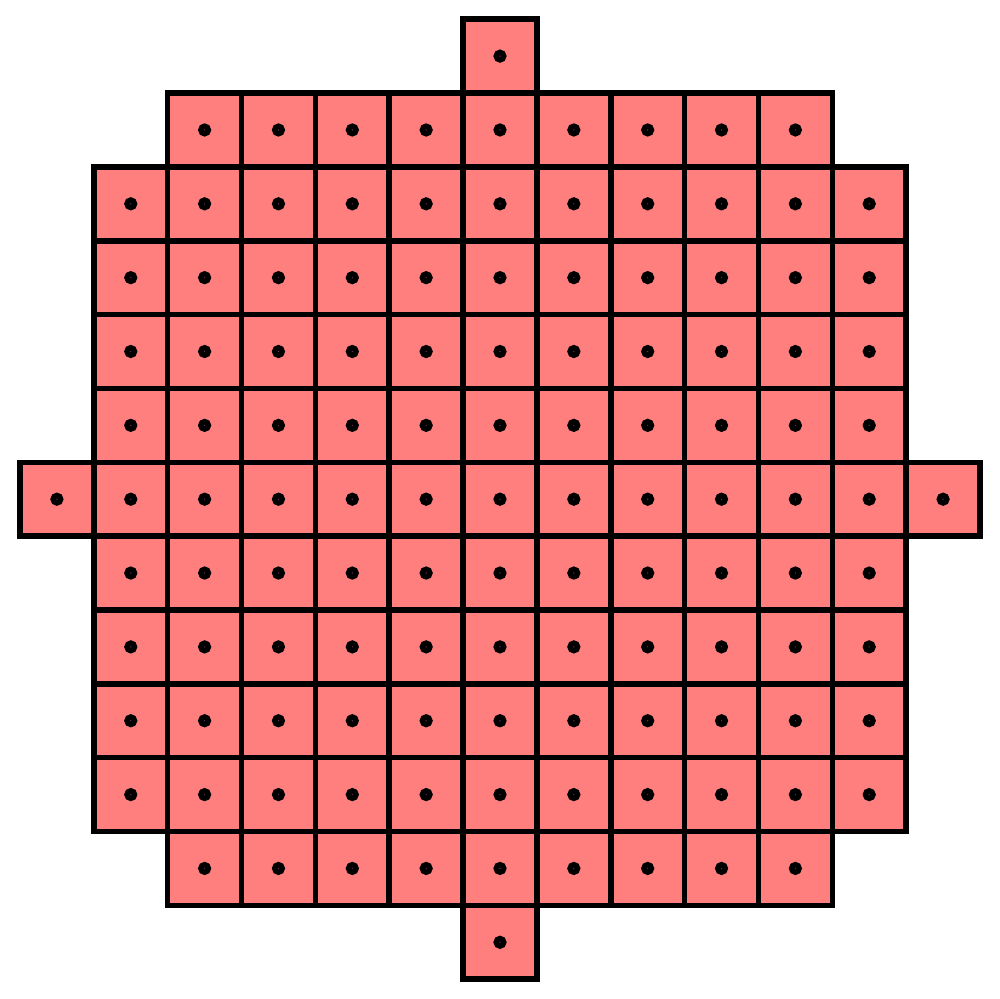}
	\end{minipage}
	\begin{minipage}[b]{0.15\linewidth}
		\centering
		\includegraphics[scale=0.17]{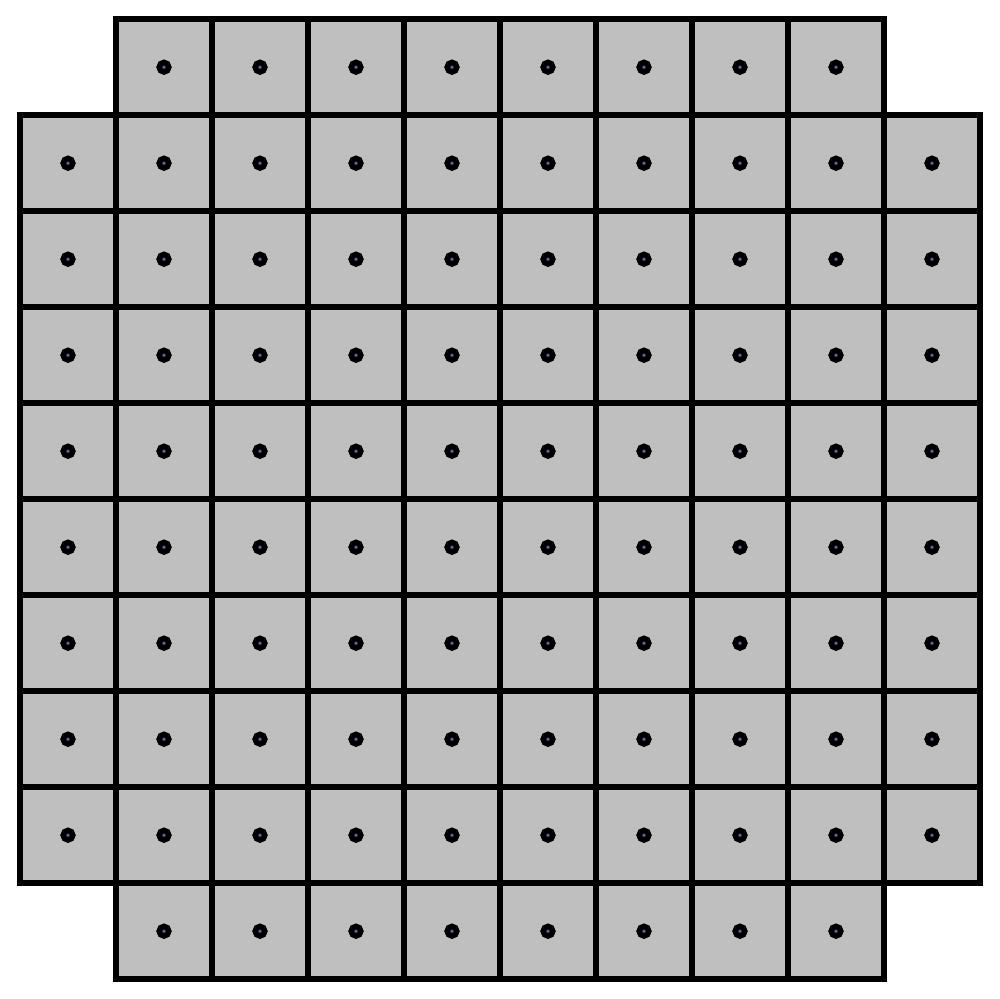}
	\end{minipage}
	\begin{minipage}[b]{0.15\linewidth}
		\centering
		\includegraphics[scale=0.17]{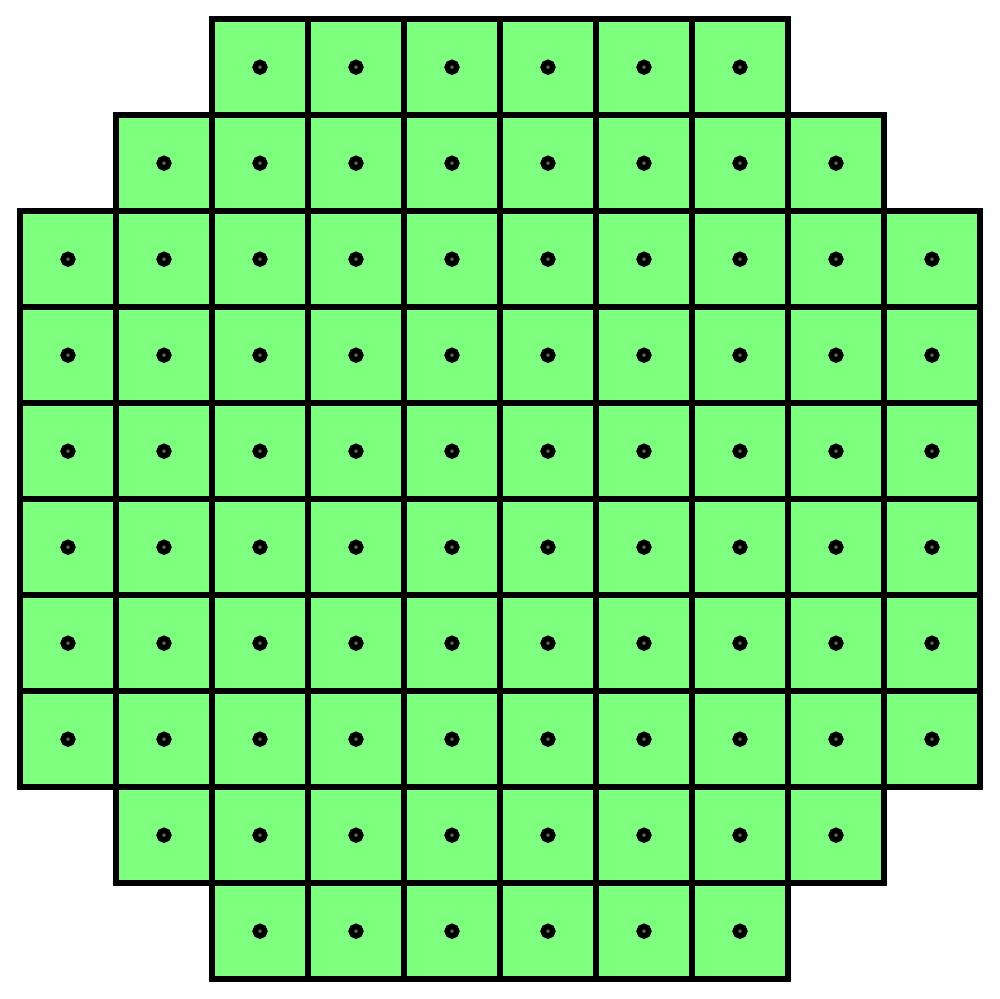}
	\end{minipage}
\caption{From the left to the right are illustrated Items (i), (ii), (iii) and (iv) of Proposition \ref{pontos} ($T_4^{2}(6)$, $T_3^{2}(6)$, $T_3^{2}(5.5)$ and $T_{2}^{2}(5.5)$, respectively).}
	\label{Fig2}
\end{figure}

\section{Perfect and quasi-perfect codes}\label{Quasi-Perfect}

It is well known that the packing radius of a code $\mathscr{C} \subset \mathbb{Z}^n$ in the $\ell_1$ metric depends only on the minimum distance of the code and it is given by the formula
$r_1 = r_1(\mathscr{C}) = \left\lfloor \frac{d_{1}(\mathscr{C})-1}{2} \right\rfloor.$ In this case a code $\mathscr{C}$ is called perfect if $\bigcup_{\cc \in \mathscr{C}} \cc + B_{1}^{n}(r_1) = \mathbb{Z}^{n}$ (or $\mathbb{Z}_q^{n}$) and $B_{1}^{n}(\cc_1,r_1) \cap B_{1}^{n}(\cc_2,r_1) = \emptyset$ for all $\cc_1, \cc_2 \in \mathscr{C}$ with $\cc_1 \neq \cc_2$. A code $\mathscr{C}$ is called quasi-perfect if $\bigcup_{\cc \in \mathscr{C}} \cc + B_{1}^{n}(r_1) \subsetneq \mathbb{Z}^{n}$ (or $\mathbb{Z}_q^{n}$) and $\bigcup_{\cc \in \mathscr{C}} \cc + B_{1}^{n}(r_1+1) = \mathbb{Z}^{n}$.

For the $\ell_p$ metric, $2 \leq p < \infty$, two codes with same minimum distance may have different packing radii and thus the packing radius is not uniquely determined by the minimum distance (see \cite[Remark 5.1]{PerfectpLee}). Moreover, the packing radius in the $\ell_p$ metric for $2 \leq p < \infty$ is not necessarily an integer number (as it can be seen next in Example \ref{radius}). 

In order to define the packing and the covering radii of $\mathscr{C} \subseteq \mathbb{Z}^{n}$ in the $\ell_p$ metric for $2\leq p < \infty$  we first define the {distance set} of the $\ell_p$ metric in $\mathbb{Z}^{n}$ as $$\mathcal{D}_{p,n} = \{d \in \mathbb{R}, \mbox{ such that there are }  z \in \mathbb{Z}^{n} \mbox{ and } c \in \mathscr{C} \mbox{ with } d_p(z,c) = d\}.$$ It follows that$D_{p,n} \subset \left\{0, 1^{1/p}, 2^{1/p}, 3^{1/p},\ldots \right\}$. 

\begin{example}\label{radius} Consider $\Lambda \subseteq \mathbb{Z}^{n}$ the lattice generated by $\{(5,11),(13,1)\}$. The first elements of $D_{2,n}$ are $1^{1/2}, 2^{1/2}, 4^{1/2}, 5^{1/2}, 8^{1/2}, 9^{1/2}, 10^{1/2}, 13^{1/2}, ,$ \linebreak $ 16^{1/2},$  $17^{1/2},$  $18^{1/2}, 20^{1/2}, 25^{1/2}, 26^{1/2}, 29^{1/2}, 32^{1/2}, 34^{1/2}, 36^{1/2}, 37^{1/2}$, $40^{1/2}$, $41^{1/2}, 45^{1/2}$, $49^{1/2}$ and $50^{1/2}$. In this case, $r_2 = 37^{1/2}$ is the largest value of $D_{2,n}$ such that the balls centered at the points of $\Lambda$ with radius $r_2$ do not intercept each other and $R_2 = 50^{1/2}$ is the smallest value of $D_{2,n}$ such that the union of the balls centered at the points of $\Lambda$ with radius $R_2$ covers $\mathbb{Z}^{n}$.

\begin{figure}[h!]
	\begin{minipage}[b]{0.30\linewidth}
		\centering
		\includegraphics[scale=0.3]{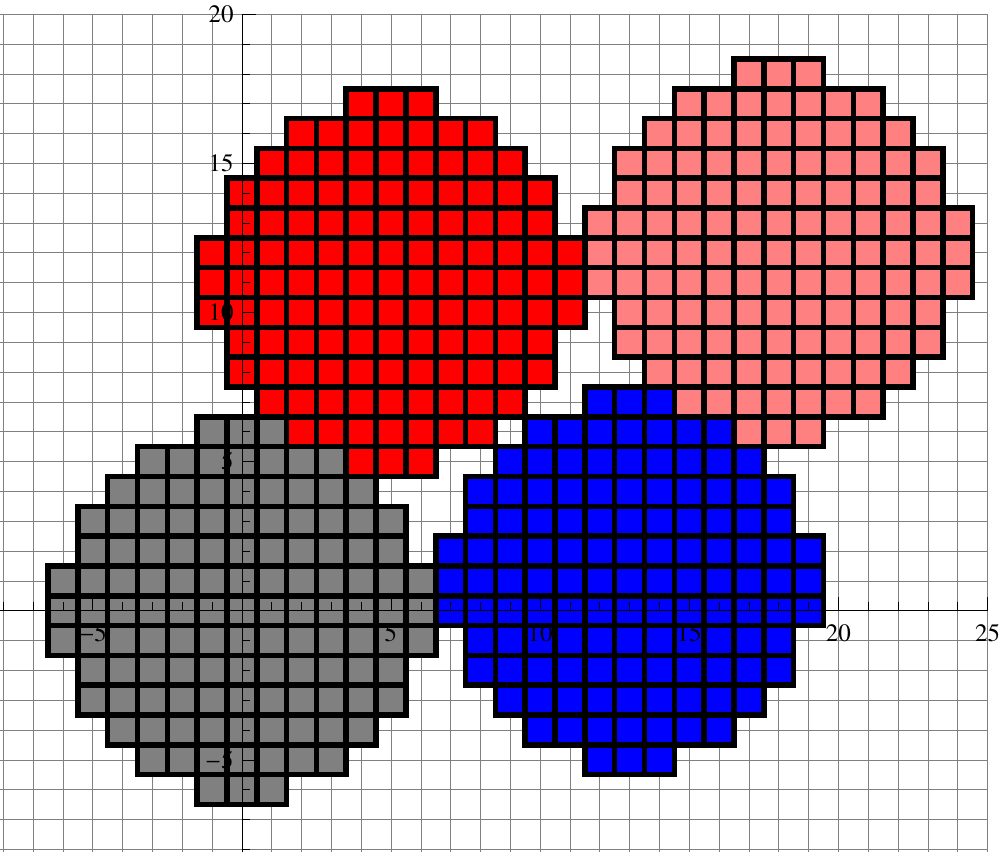}
	\end{minipage}
	\begin{minipage}[b]{0.30\linewidth}
		\centering
		\includegraphics[scale=0.3]{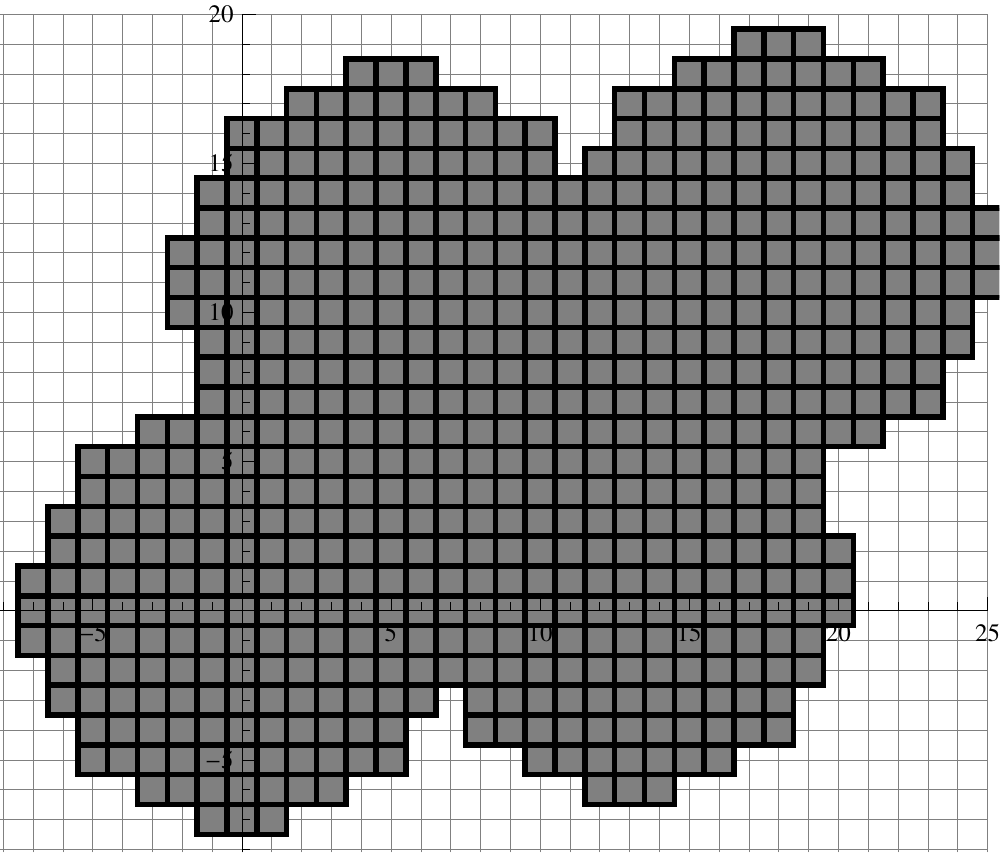}
	\end{minipage}
\caption{From the left to the right some polyominoes associated to balls in the $\ell_2$ metric centered at points of $\Lambda$ with radii $r_2$ and $R_2$, respectively.}
	\label{Fig3}
\end{figure}
\end{example}

The {\it packing radius} of a code $\mathscr{C} \subseteq \mathbb{Z}^{n}$ in the $\ell_p$ metric is the greatest $r \in \mathcal{D}_{p,n}$ such that $B_{p}^n(\bm{\xx},r) \cap B_{p}^n(\bm{\yy},r) = \emptyset$ holds for all $\xx, \yy \in \mathscr{C}$. The packing radius of a code $\mathscr{C} \subset \mathbb{Z}^{n}$ in the $\ell_p$ metric will be denoted by $r_p=r_p(\mathscr{C})$.

The {\it covering radius} of a code $\mathscr{C} \subseteq \mathbb{Z}^{n}$ in the $\ell_p$ metric is the smallest $r \in \mathcal{D}_{p,n}$ such that $\bigcup_{\cc \in \mathscr{C}} \cc+ B_p^{n}(,r) = \mathbb{Z}^{n}$. The covering radius of a code $\mathscr{C} \subset \mathbb{Z}^{n}$in the $\ell_p$ metric will be denoted by $R_p=R_p(\mathscr{C})$.

We denote by $\overline{r_p} = \overline{r_p}(\Lambda)$ and $\overline{R_p}= \overline{R_p}(\Lambda)$, respectively, the real packing and covering radii of a lattice $\Lambda$ in $\mathbb{R}^{n}$ (for the packing radius the balls centered at $\Lambda$ with radius $\overline{r_p}$ do not intercept each other in $\mathbb{R}^{n}$ and for the covering radius the union of the balls centered at the points of $\Lambda$ with radius $\overline{R_p}$ covers $\mathbb{R}^{n}$).  

\begin{example} Consider the lattice $\Lambda$ generated by $\{(1,4),(0,24)\}$. We have that  $r_p = 2$, $\overline{r_p} = 2.0616$,  $R_p = 3.1623$ and $\overline{R_p} = 3.3001$. 

\begin{figure}[h!]
	\begin{minipage}[b]{0.2\linewidth}
		\centering
		\includegraphics[scale=0.2]{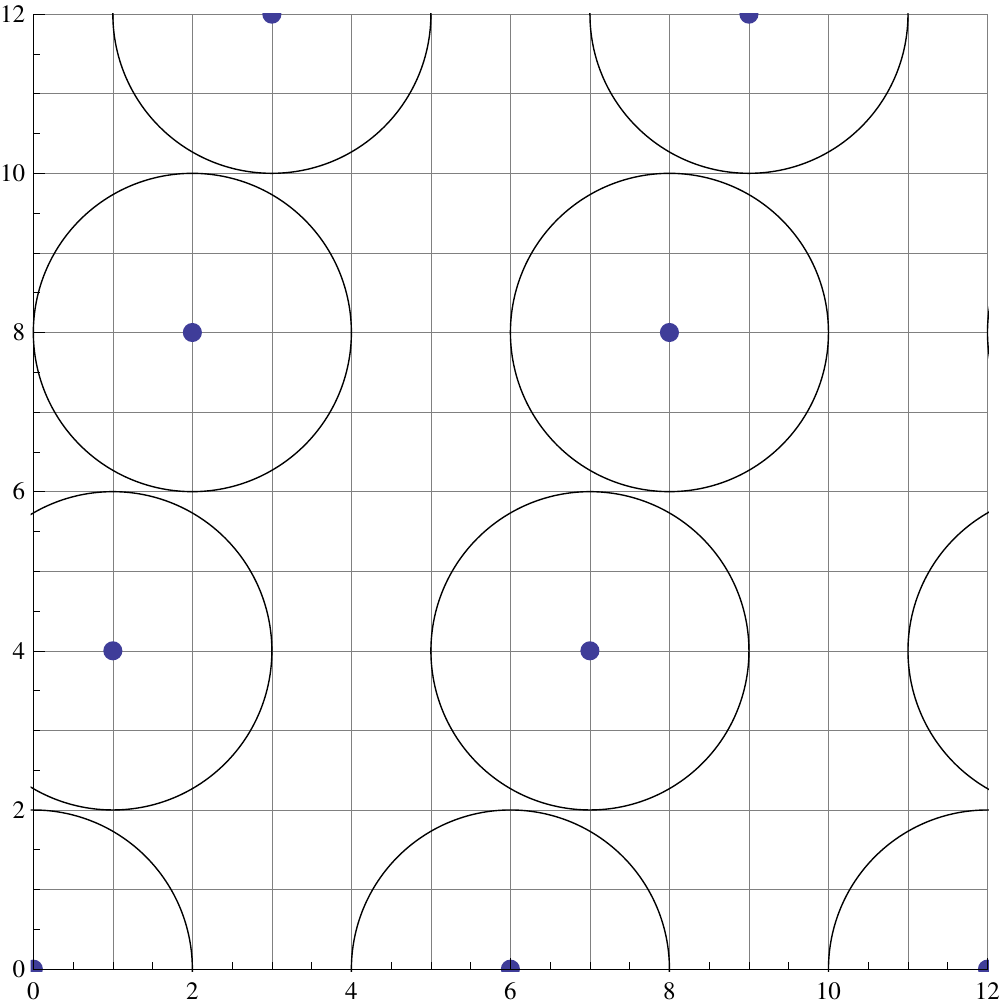}
	\end{minipage}
	\begin{minipage}[b]{0.25\linewidth}
		\centering
		\includegraphics[scale=0.2]{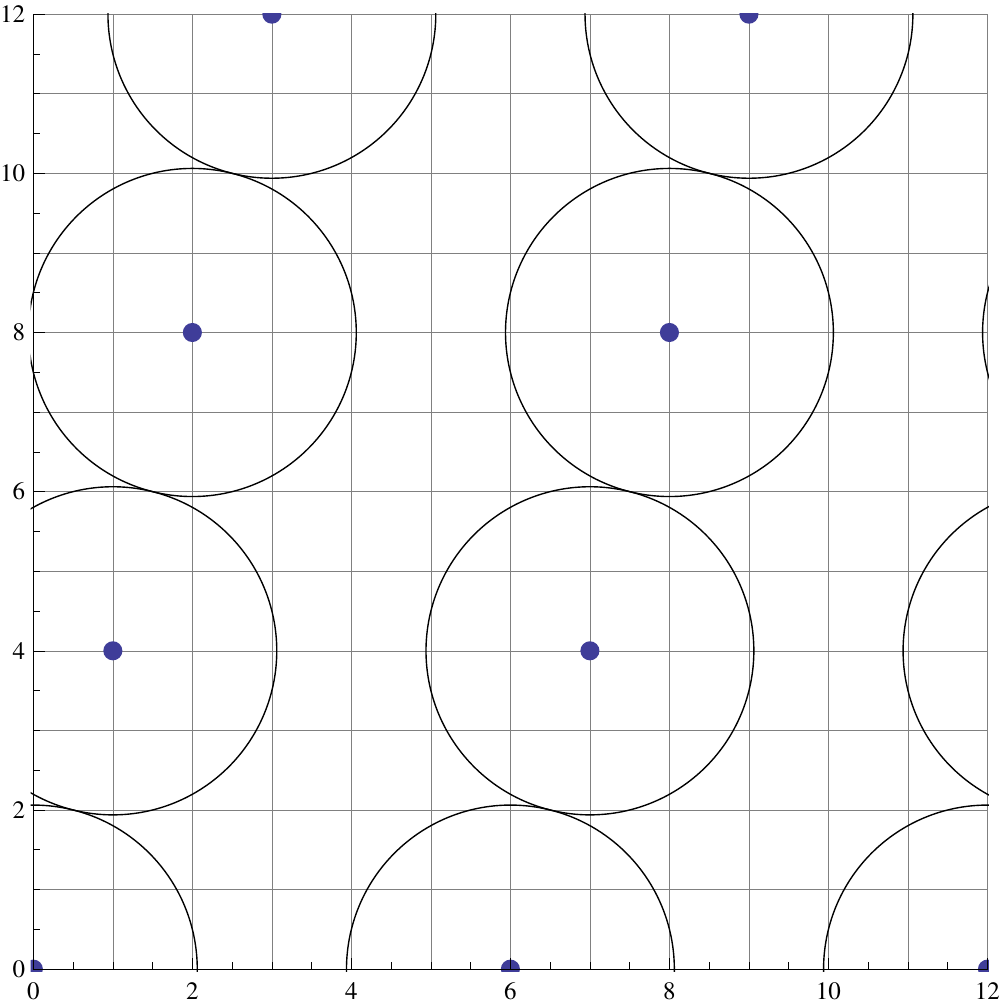}
	\end{minipage}
	\begin{minipage}[b]{0.25\linewidth}
		\centering
		\includegraphics[scale=0.2]{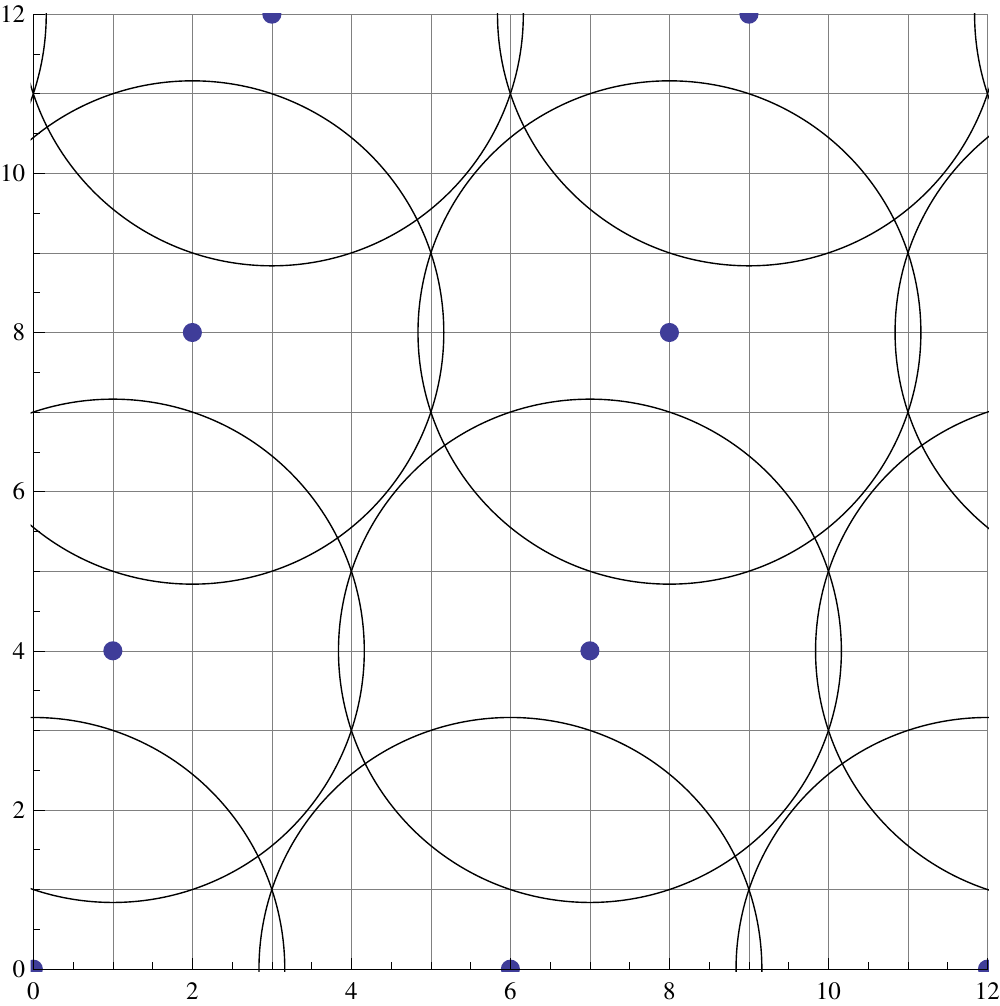}
	\end{minipage}
	\begin{minipage}[b]{0.2\linewidth}
		\centering
		\includegraphics[scale=0.2]{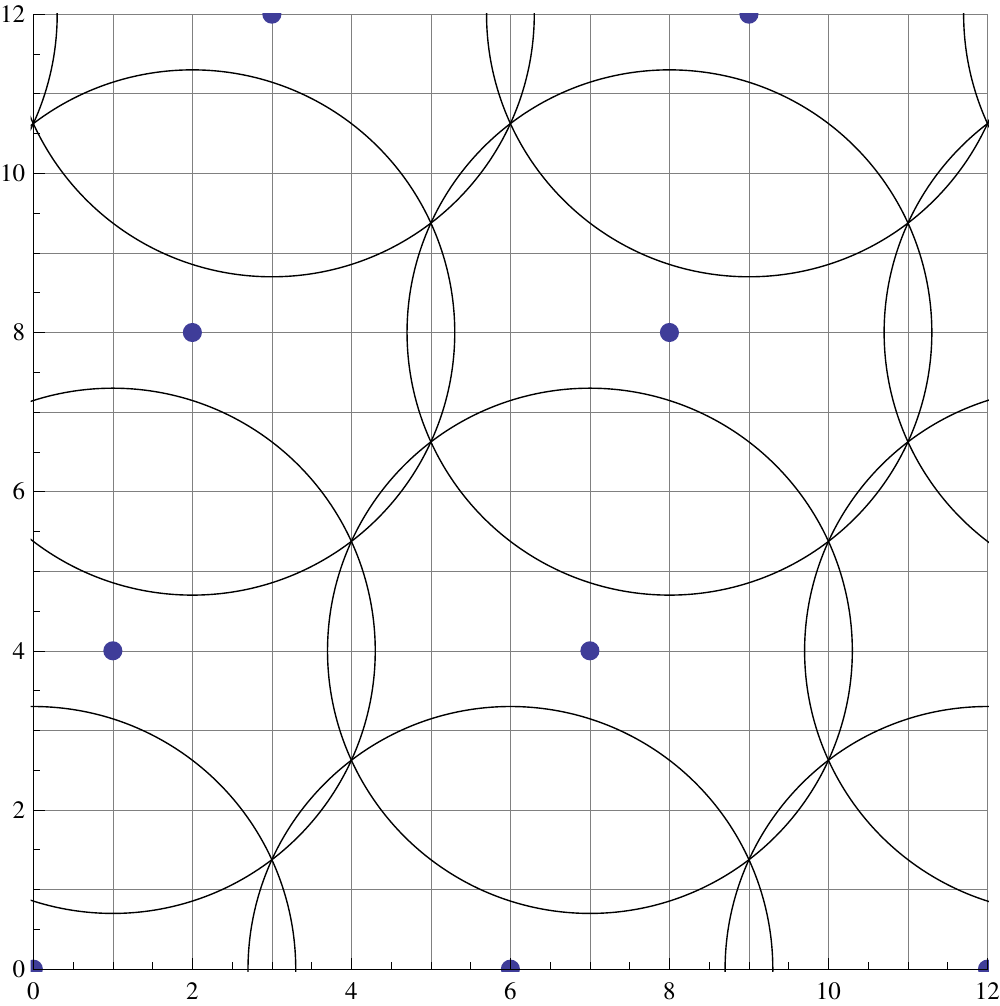}
	\end{minipage}
	\caption{From the left to the right the sphere packing with radii $r_p$, $\overline{r_p}$, $R_p$ and $\overline{R_p}$, respectively.}
	\label{Fig4}
\end{figure}
\end{example}

We define the distance of two elements $r_a, r_b \in D_{p,n}$ with $r_a < r_b$ as $d(r_a,r_b) = \#\left( D_{p,n} \cap [r_a,r_b)\right)$, where $[r_a,r_b)$ denotes the closed interval in $\mathbb{R}$ and $d(r_a,r_a) = 0$. We say that a lattice $\Lambda$ is {\it $t$-imperfect} if $d(r_p,R_p) = t$. When $t = 0$, that is $r_p = R_p$, the lattice is called {\it perfect}. When $t=1$ the lattice is called {\it quasi-perfect}. 

In \cite[Corollary 5.5]{PerfectpLee} it was shown that if $1 < p < \infty$ and $n \geq 2$, then the radius packing $r_p$  of a linear perfect code in the $\ell_p$ metric satisfies
	\begin{equation}\label{6}  r_p \leq \frac{n^{1/p}}{2} \frac{\left(1+(\overline{\Delta_p^{n}})^{1/n}\right)}{\left(1-(\overline{\Delta_p^{n}})^{1/n}\right)}.
	\end{equation} 
where $\overline{\Delta_p^{n}}$ denotes the supremum of the packing densities over all $n$-dimensional lattices in the $\ell_p$ metric.

We recall that the {\it covering density} of a lattice $\Lambda$ in the $\ell_p$ metric is given  by $\overline{\Theta_{p}^{n}}(\Lambda) = \frac{\mathcal{V}_{p}^{n} \overline{R_p}^n }{\det \Lambda}$, where $\mathcal{V}_{p}^{n}$ is the Euclidean volume of the $n$-dimensional unitary sphere centered at the origin in the $\ell_p$ metric \cite[p.321]{VolumBall}. $\overline{\Theta_{p}^{n}}$ denotes the  infimum  of the covering densities  over all $n$-dimensional  lattices in the $\ell_p$ metric.

Proposition \ref{bounds} give us some relations among packing and covering radii of a quasi-perfect lattice and covering density for the lattice in $\mathbb{R}^{n}$.

\begin{proposition}\label{bounds}  Let $1 < p < \infty$ and $n \geq 2$. The covering radius $R_p$ and the packing radius $r_p$ of a linear quasi-perfect code in the $\ell_p$ metric satisfies 
	\begin{equation} \label{7} \overline{\Theta _{p}^{n}} \leq \frac{\mathcal{V}_{p}^{n}\left(R_p +\frac{1}{2}\sqrt[p]{n}\right){}^n }{|{B}_p^{n}\left({r}_p\right)|}\end{equation} and
	\begin{equation} \label{8}\overline{\Theta _{p}^{n}} \leq \frac{\mathcal{V}_{p}^{n}\left(r_p +\sqrt[p]{n}\right){}^n}{|{B}_p^n\left(r_p\right)|}. \end{equation}
\end{proposition}

\begin{proof} The real and integer covering and packing radii satisfy
	$\overline{R}_p\leq R_p+\left\|\frac{1}{2}\mathbf{1}\right\|_p=R_p +\frac{1}{2}\sqrt[p]{n}$ where $\mathbf{1}=(1,\ldots ,1)\in \mathbb{Z}^n$. Therefore, it follows that
	\begin{equation*} \begin{split} \overline{\Theta_{p}^{n}} & \leq \frac{\mathcal{V}_{p}^{n}\overline{R_p}^n }{\text{vol}(\Lambda )}\leq  \frac{\mathcal{V}_{p}^{n}\left(R_p +\frac{1}{2}\sqrt[p]{n}\right)^n }{\text{vol}(\Lambda )}\leq \frac{\mathcal{V}_{n,p}\left(R_p +\frac{1}{2}\sqrt[p]{n}\right){}^n }{|{B}_p^n\left(r_p\right)|}  \end{split} \end{equation*}
	
Since the linear code is quasi-perfect by hypothesis, it follows that $\overline{R_p} \leq r_p + \sqrt[p]{n}$ and then

\begin{equation*} \begin{split} \overline{\Theta _{p}^{n}} & \leq \frac{\mathcal{V}_{p}^{n}\overline{R}_p^n }{\text{vol}(\Lambda )}\leq \frac{\mathcal{V}_{p}^{n}\left(r_p +\sqrt[p]{n}\right)^n }{\text{vol}(\Lambda )}  \leq \frac{\mathcal{V}_{p}^{n}\left(r_p +\sqrt[p]{n}\right){}^n }{|{B}_p^n\left(r_p\right)|} \end{split}\end{equation*}
\end{proof}
Since neither the packing and covering density of any lattice cannot exceed the best possible density in dimension $n$, these inequalities give a limitation on the packing radius of a quasi-perfect code.  
A numerical comparison among these limitations, provided by (\ref{6}), (\ref{7}) and (\ref{8}) is presented in Table 1. On the left table, the third and fourth columns represent bounds for the packing and covering densities of the perfect lattices with packing radius $r_2 = R_2$ obtained from Inequalities (\ref{6}) and (\ref{7}), respectively. According to Inequality (\ref{7}) we must test lattices with packing radius upper bounded by $49$ while inequality (\ref{6}) shows the packing radius must be upper bounded by $833$, otherwise we had packing and covering densities greatest and smallest that the maximum and minimum possible values ($0,9069$ e $1.2092$, e.g., \cite{SloaneLivro})
 in this dimension, respectively. On the right table, the third, fourth and fifth columns represent bounds for the packing and covering densities of the quasi-perfect lattices with packing radius $r_2$ obtained from Inequalities (\ref{6}), (\ref{7}) and (\ref{8}), respectively. From Inequality (\ref{7}) the packing radius must be upper bounded by $74$ while Inequality (\ref{8}) it must be upper bounded by $196$. In this case Inequality (\ref{7}) is more appropriated.


\begin{table}[h]
\small
\begin{tabular}{|c|c|c|c|}
\hline & & & \\
 $r_2^2$ & $\mu(2,2,r_2)$ & $\overline{\Delta_2^2}\geq$  & $\overline{\Theta _2^2}\leq$  \\ \hline \hline
 41 & 137 & 0.6418 & 1.1593 \\
 45 & 145 & 0.6549 & 1.1914 \\
 \bf 49 & 149 & 0.6667 & \bf 1.2524 \\
 50 & 161 & 0.6694 & 1.1805 \\
 52 & 169 & 0.6747 & 1.1655 \\
 829 & 2601 & 0.9064 & 1.0511 \\
 832 & 2609 & 0.9066 & 1.0516 \\
 \bf 833 & 2617 & \bf 0.9066 & 1.0496 \\
 841 & 2629 & 0.9071 & 1.0546 \\
 842 & 2637 & 0.9071 & 1.0526 \\ \hline
 \end{tabular}
\hfill
\small
\begin{tabular}{|c|c|c|c|c|}
\hline  & & & & \\
 $r_2^2$ & $\mu(2,2,r_2)$ & $\overline{\Delta_2^2}\geq$  & $\overline{\Theta _2^2}\leq$  & $\overline{\Theta _2^2}\leq$  \\ \hline \hline
 72 & 225 & 0.716 & 1.195 & 1.3683 \\
 73 & 233 & 0.7176 & 1.1685 & 1.3371 \\
 \bf 74 & 241 & 0.7193 & \bf 1.2143 & 1.3079 \\
 80 & 249 & 0.7284 & 1.1889 & 1.3538 \\
 81 & 253 & 0.7298 & 1.1835 & 1.3467 \\
 193 & 601 & 0.8156 & 1.1197 & 1.2247 \\
 194 & 609 & 0.8161 & 1.1158 & 1.2143 \\
 \bf 196 & 613 & 0.8169 & 1.1139 & \bf 1.2177 \\
 197 & 621 & 0.8174 & 1.1155 & 1.2076 \\
 200 & 633 & 0.8186 & 1.1048 & 1.2011 \\ \hline
\end{tabular}
\label{comparison2}
\caption{Some bounds obtained from Inequalities (\ref{6}), (\ref{7}) and (\ref{8}).}
\end{table}

The {\it discrete packing density} of a lattice $\Lambda$ in $\mathbb{Z}^{n}$ in the $\ell_p$ metric is given by $\Delta_p^{n}(\Lambda) = \frac{\mu(n,p,r_p)}{det \Lambda}$. The {\it discrete covering density} of a lattice $\Lambda$ in $\mathbb{Z}^{n}$ is given by $\Theta_p^{n}(\Lambda) = \frac{\mu(n,p,R_p)}{det \Lambda}.$

\begin{example} Table \ref{tab:24} shows some parameters of all integer lattices with volume $M=24$ in the $\ell_2$ metric. Here $\Delta_p = \Delta_p^{2}(\Lambda)$,  $\overline{\Delta_p} = \overline{\Delta_p^{2}}(\Lambda)$,  $\Theta_p = \Theta_p^{2}(\Lambda)$ and  $\overline{\Theta_p} = \overline{\Theta_p^{2}}(\Lambda)$ denote discrete packing density, packing density, discrete covering density and covering density of the lattice $\Lambda$, respectively. In this case we have only one quasi-perfect lattice up to congruence.
\end{example}
	\begin{table}[!htb]
\begin{footnotesize}
\begin{tabular}{|cccccccccc|} \hline
Lattice & $t$ & $r_2$ & $\overline{r_2}$ & $R_2$ & $\overline{R_2}$ & $\Delta_2$ & $\overline{\Delta_2^2}$ & $\Theta_2^2$ & $\overline{\Theta _2^2}$ \\ \hline

$ \left(
\begin{smallmatrix}
 1 & 0 \\
 0 & 24 \\
\end{smallmatrix}
\right)$ & 58 & 0 & 0.5 & 12 & 12.0104 & 0.0417 & 0.0327 & 18.375 & 18.8823 \\
$ \left(
\begin{smallmatrix}
 1 & 1 \\
 0 & 24 \\
\end{smallmatrix}
\right)$ & 32 & 0 & 0.7071 & 8.4853 & 8.5147 & 0.0417 & 0.0654 & 9.375 & 9.4902 \\
$ \left(
\begin{smallmatrix}
 1 & 2 \\
 0 & 24 \\
\end{smallmatrix}
\right)$ & 14 & 1 & 1.118 & 5.3852 & 5.4271 & 0.2083 & 0.1636 & 4.0417 & 3.8554 \\
 $ \left( 
\begin{smallmatrix}
 1 & 3 \\
 0 & 24 \\
\end{smallmatrix}
\right) $ & 7 & 1.4142 & 1.5811 & 4 & 4.0139 & 0.375 & 0.3272 & 2.0417 & 2.1089 \\
 $ \left( 
\begin{smallmatrix}
 1 & 4 \\
 0 & 24 \\
\end{smallmatrix}
\right) $ & 4 & 2 & 2.0616 & 3.1623 & 3.3001 & 0.5417 & 0.5563 & 1.5417 & 1.4256 \\
 $\boldmath{ \left( 
\begin{smallmatrix}
 1 & 5 \\
 0 & 24 \\
\end{smallmatrix}
\right)} $ & \bf 1 & \bf 2.2361 & \bf 2.5495 & \bf 2.8284 & \bf 3.0641 & \bf 0.875 & \bf 0.8508 & \bf 1.0417 & \bf 1.229 \\
 $ \left( 
\begin{smallmatrix}
 1 & 6 \\
 0 & 24 \\
\end{smallmatrix}
\right) $ & 5 & 1.4142 & 2 & 3.1623 & 3.4004 & 0.375 & 0.5236 & 1.5417 & 1.5135 \\
 $ \left( 
\begin{smallmatrix}
 1 & 7 \\
 0 & 24 \\
\end{smallmatrix}
\right) $ & 4 & 2 & 2.1213 & 3.1623 & 3.5355 & 0.5417 & 0.589 & 1.5417 & 1.6362 \\
 $ \left( 
\begin{smallmatrix}
 1 & 8 \\
 0 & 24 \\
\end{smallmatrix}
\right) $ & 8 & 1.4142 & 1.5 & 4.1231 & 4.1552 & 0.375 & 0.2945 & 2.375 & 2.2601 \\
 $ \left( 
\begin{smallmatrix}
 1 & 9 \\
 0 & 24 \\
\end{smallmatrix}
\right) $ & 4 & 2 & 2.1213 & 3.1623 & 3.2596 & 0.5417 & 0.589 & 1.5417 & 1.3908 \\
 $ \left( 
\begin{smallmatrix}
 1 & 10 \\
 0 & 24 \\
\end{smallmatrix}
\right) $ & 4 & 2 & 2.2361 & 3.1623 & 3.3657 & 0.5417 & 0.6545 & 1.5417 & 1.4828 \\
 $ \left( 
\begin{smallmatrix}
 1 & 11 \\
 0 & 24 \\
\end{smallmatrix}
\right) $ & 10 & 1 & 1.4142 & 4.2426 & 4.3605 & 0.2083 & 0.2618 & 2.5417 & 2.4889 \\
 $ \left( 
\begin{smallmatrix}
 1 & 12 \\
 0 & 24 \\
\end{smallmatrix}
\right) $ & 18 & 0 & 1 & 6 & 6.0417 & 0.0417 & 0.1309 & 4.7083 & 4.7781 \\
 $ \left( 
\begin{smallmatrix}
 2 & 0 \\
 0 & 12 \\
\end{smallmatrix}
\right) $ & 19 & 0 & 1 & 6.0828 & 6.0828 & 0.0417 & 0.1309 & 5.0417 & 4.8433 \\
 $ \left( 
\begin{smallmatrix}
 2 & 2 \\
 0 & 12 \\
\end{smallmatrix}
\right) $ & 11 & 1 & 1.4142 & 4.4721 & 4.4721 & 0.2083 & 0.2618 & 2.875 & 2.618 \\
 $ \left( 
\begin{smallmatrix}
 2 & 3 \\
 0 & 12 \\
\end{smallmatrix}
\right) $ & 5 & 2 & 1.8028 & 3.6056 & 3.6336 & 0.5417 & 0.4254 & 1.875 & 1.7283 \\
 $ \left( 
\begin{smallmatrix}
 2 & 4 \\
 0 & 12 \\
\end{smallmatrix}
\right) $ & 4 & 2 & 2.2361 & 3.1623 & 3.1623 & 0.5417 & 0.6545 & 1.5417 & 1.309 \\
 $ \left( 
\begin{smallmatrix}
 2 & 6 \\
 0 & 12 \\
\end{smallmatrix}
\right) $ & 5 & 1.4142 & 2 & 3.1623 & 3.3333 & 0.375 & 0.5236 & 1.5417 & 1.4544 \\
 $ \left( 
\begin{smallmatrix}
 3 & 0 \\
 0 & 8 \\
\end{smallmatrix}
\right) $ & 8 & 1.4142 & 1.5 & 4.1231 & 4.272 & 0.375 & 0.2945 & 2.375 & 2.3889 \\
 $ \left( 
\begin{smallmatrix}
 3 & 4 \\
 0 & 8 \\
\end{smallmatrix}
\right) $ & 2 & 2.2361 & 2.5 & 3 & 3.125 & 0.875 & 0.8181 & 1.2083 & 1.2783 \\
 $ \left( 
\begin{smallmatrix}
 4 & 0 \\
 0 & 6 \\
\end{smallmatrix}
\right) $ & 6 & 1.4142 & 2 & 3.6056 & 3.6056 & 0.375 & 0.5236 & 1.875 & 1.7017 \\ \hline
\end{tabular}
\end{footnotesize}
\label{tab:24}
\caption{Codes in $\mathbb{Z}^2$, $p = 2$ and their respective degree of imperfection $t$.}
\end{table}

\section{Families of $t$-imperfect lattices in the $\ell_p$ metric} \label{ExplicitConstruction}

In this section we present some families of lattices  in the $\ell_p$ metric and calculate their imperfection degrees and discrete packing densities. In some cases the imperfection degree ranges with the packing radius.

\begin{proposition} \label{FormatA} Let $r>1$ an integer and $p$ an integer such that $p \geq \frac{\ln 2}{\ln
\left(\frac{r}{r -1}\right)}$. The lattice $\Lambda_r$ with basis $\{(r, 2r-1),(2r,-1)\}$ is quasi-perfect for $r=2$ and $r=3$ and $(r-2)$-imperfect for $r \geq 3$. It has discrete packing density
$\Delta_p^{2}(\Lambda_r) = \displaystyle\frac{(2r-1)^2+4}{4r^{2}-r}$. 
\end{proposition}
\begin{proof} $\Lambda_r$ is not perfect because $(r,r-2) \not\in \bigcup_{\cc \in \Lambda} \cc+ B_{p}^{2}(r)$. If $r_1 = \sqrt[p]{r^{p}+1}$, then $r_1 \in D_{p,2}$, $r \leq r_1$ and $(r,r-2) \in  \bigcup_{\cc \in \Lambda_r} \cc+ B_{p}^{2}(r_1)$ if and only if $r \in \{2,3\}$. If $r\geq 3$, then $\sqrt[p]{r^p}<\sqrt[p]{r^p+1}<\sqrt[p]{r^p+2^p}<\ldots<\sqrt[p]{r^p+(r-2)^p}<\sqrt[p]{(r+1)^p}.$ The last inequality follows from the fact $p \geq \frac{\ln 2}{\ln
\left(\frac{r}{r -1}\right)}$ implies $2(r-1)^{p} \leq r^{p}$ and that $x^{p}$ is a concave upward function when $p>2$ and $x>0$. Indeed, a convex combination of the image of $r-2$ and	
$r+1$ is greather than the image of $r$, that is,
\[\begin{split} & \frac{1}{2}(r-1)^p+\frac{1}{2}(r+1)^p \geq  r^p\Rightarrow (r+1)^p\geq 2r^p-(r-1)^p=r^p+r^p-(r-1)^p \\& {\geq }r^p +2 (r-1)^p-(r-1)^p=r^p+(r-1)^p\geq r^p+(r-2)^p. \end{split}\]	
From Proposition \ref{pontos}, $\mu(2,p,r) = (2r-1)^{2}+4$ and since $\det(\Lambda_r) = 4r^{2}-r$, the discrete packing density is $\Delta_p^{2}(\Lambda_r) = \displaystyle\frac{(2r-1)^2+4}{4r^{2}-r}$. \end{proof}

\begin{table}[h!]
	\caption{The minimum value of $p$ such that  $\frac{\ln 2}{\ln \left(\frac{r}{r -1}\right)} \leq p$.} 
	\centering 
	\begin{tabular}{|c|c|c|c|c|c|c|c|c|c|c|c|c|c|c|} 
		\hline 
	$r$ & $2$ & $3$ & $4$ & $5$ & $6$ & $7$ & $8$ & $9$ & $10$ & $11$ & $12$ & $13$ & $14$  \\  
	\hline 
		$p$ & $1$ & $2$ & $3$ & $4$ & $4$ & $5$ & $6$ & $6$ & $7$ & $8$ & $8$ & $9$ & $10$ 
\\		
		 \hline

\end{tabular}
	\label{table1} 
\end{table}

\begin{example} The lattice $\Lambda_3$ with basis $\{(3,5),(6,-1)\}$  described in Proposition \ref{FormatA} is quasi-perfect in the $\ell_p$-metric for $p \geq 2$. $\Lambda_3$ is congruent to the lattice with basis $\{(1,6), (0,33)\}$ described in  the set ${\mathcal A}$ in Section \ref{Algorithm}. 
The lattice $\Lambda_4$ with basis  $\{(4,7),(8,-1)\}$  described in Proposition \ref{FormatA} is $2$-imperfect in the $\ell_p$-metric for $p \geq 3$. 
	 \begin{figure}[h!]
	 	\begin{minipage}[b]{0.2\linewidth}
	 		\centering
	 		\includegraphics[scale=0.2]{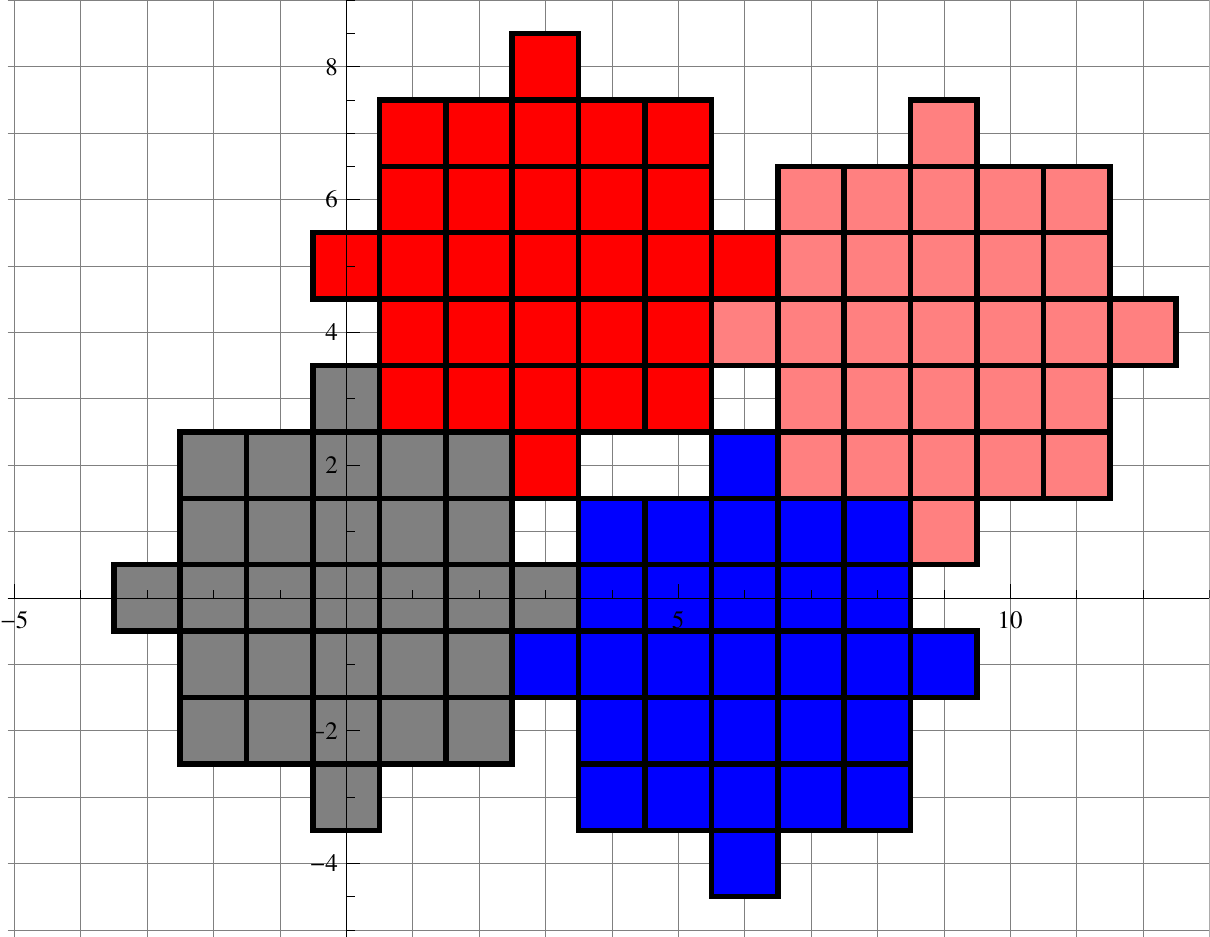}
	 	\end{minipage}
	 	\begin{minipage}[b]{0.25\linewidth}
	 		\centering
	 		\includegraphics[scale=0.2]{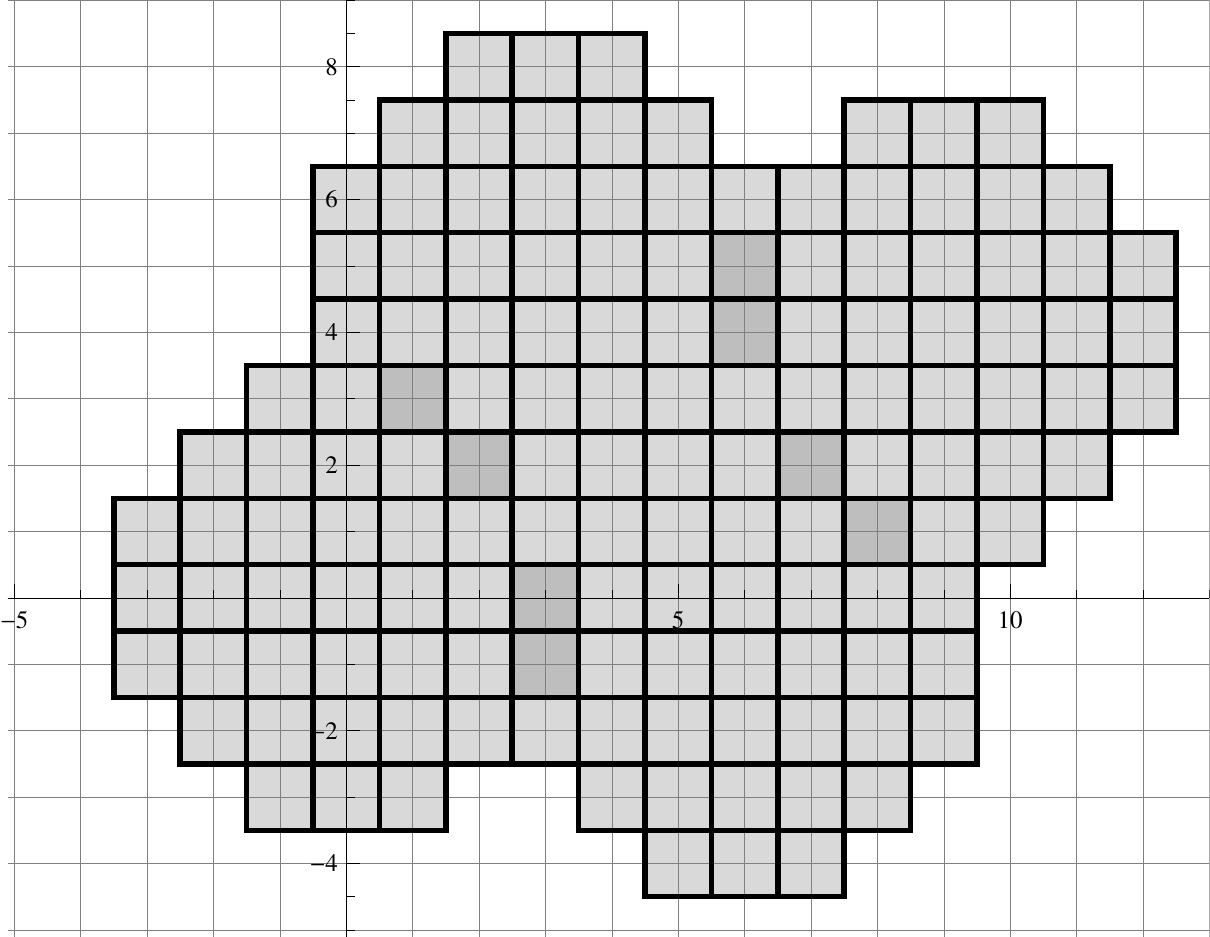}
	 	\end{minipage}
	 	\begin{minipage}[b]{0.25\linewidth}
	 		\centering
	 		\includegraphics[scale=0.2]{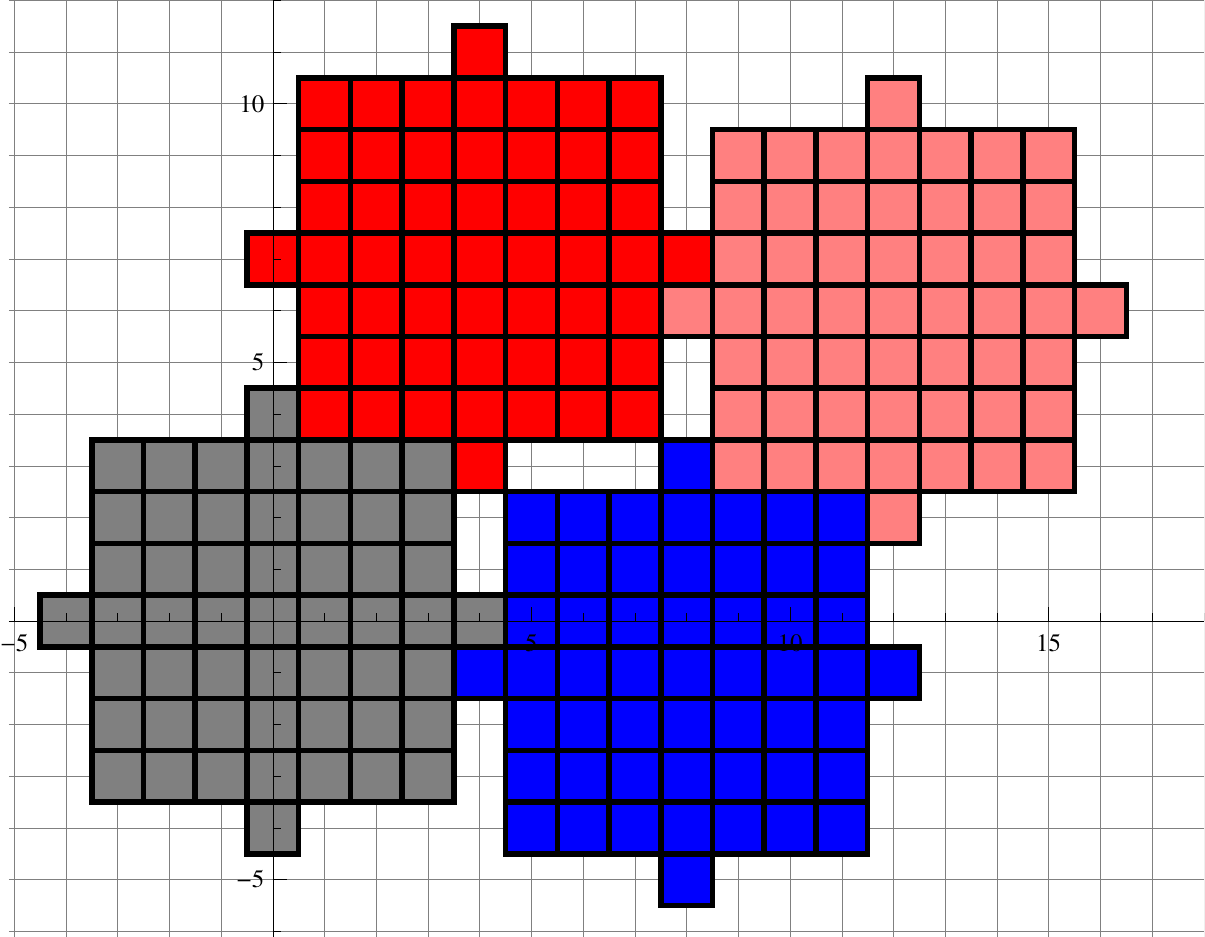}
	 	\end{minipage}
	 	\begin{minipage}[b]{0.2\linewidth}
	 		\centering
	 		\includegraphics[scale=0.2]{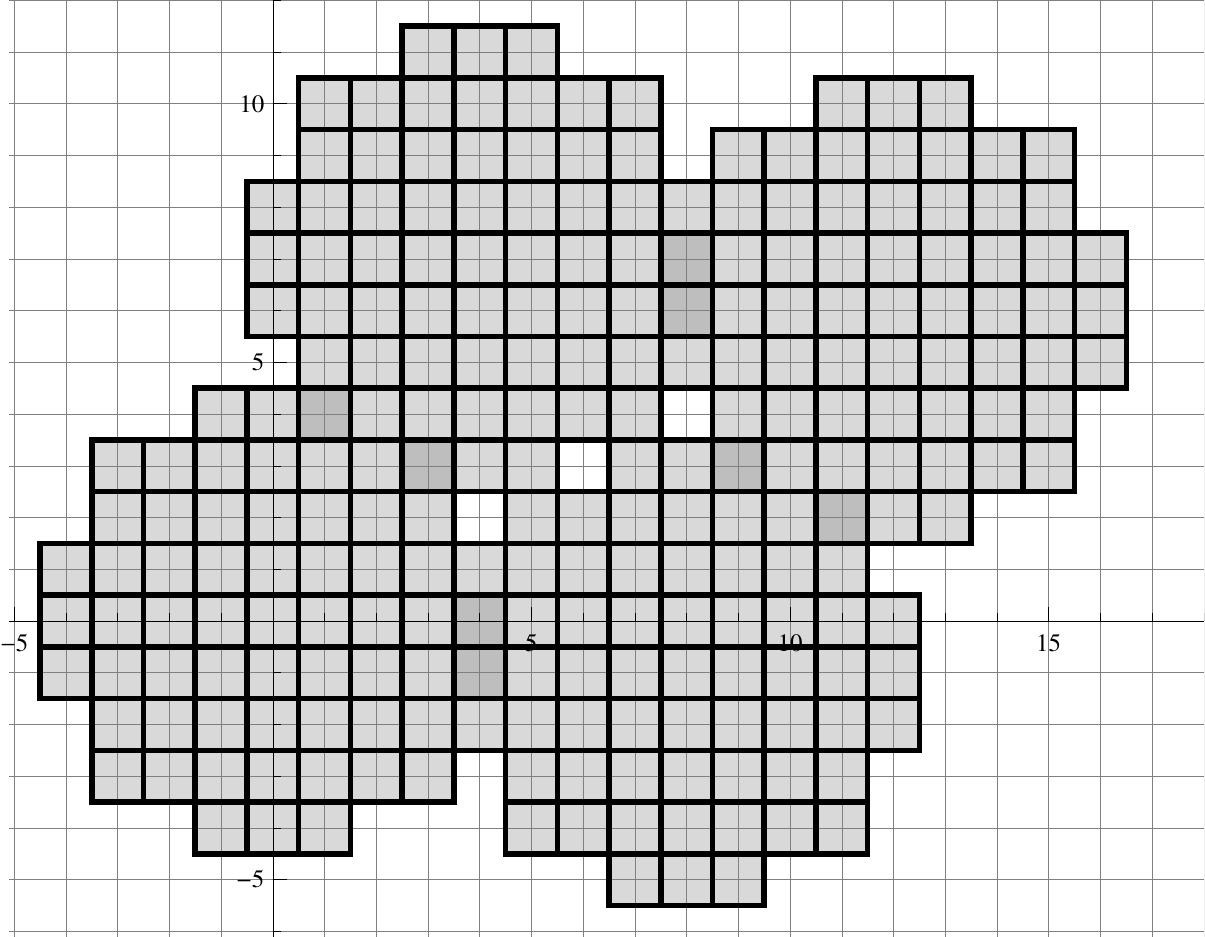}
	 	\end{minipage}
	 	\caption{From the left to the right the polyominoes associated to the balls $B_p^{2}(r)$ for $(r,p) = (3,2)$, $(5^{1/2},2)$, $(4,3)$ and $((4^{3}+1)^{1/3},3)$, respectively.}
	 	\label{Fig4}
	 \end{figure}
	 \end{example}

\begin{proposition}\label{FormatB} Let $r$ be an integer, $p< \frac{\ln 2}{\ln
		\left(\frac{r}{r -1}\right)}$ and $(r-1)^{p} + (r-2)^{p} \leq r^{p}$. The lattice $\Lambda_r$ with basis $\{(r-1, 2r-1),(2r,-1)\}$ is $(r-1)$-imperfect. It has discrete packing density $\Delta_p^{2}(\Lambda_r) = \displaystyle\frac{(2r-1)^{2}}{4r^{2}-r-1}$. 
\end{proposition}
\begin{proof} By hypothesis $2(r-1)^p > r^p \geq (r-1)^p+(r-2)^p.$ Then
$\sqrt[p]{r^p}<\sqrt[p]{r^p+1}<\sqrt[p]{r^p+2^p}<\ldots<\sqrt[p]{r^p+(r-2)^p}$. We also have $\sqrt[p]{2(r-1)^p}< \sqrt[p]{r^p+(r-2)^p}$ since the ball is convex. Moreover,
$r^p+i^p=2(r-1)^p$ for $i=0,1,\ldots,r-2$ has no integer solution. Hence, the radius $r+1$ covers the point $(r,r-2)$. The formula for the packing density follows from Proposition \ref{pontos} since $\mu(2,p,r) = (2r-1)^{2}$. 
\end{proof}

\begin{table}[h!]
	\caption{Values of $p$ such that  $\frac{\ln 2}{\ln \left(\frac{r}{r -1}\right)} \leq p$ and $(r-1)^{p} + (r-2)^{p} \leq r^{p}$.} 
	\centering 
	\begin{tabular}{|c|c|c|c|c|c|c|c|c|c|c|c|c|c|} 
		\hline 
		$r$ &  $3$ & $4$ & $5$ & $6$ & $7$ & $8$ & $9$ & $10$ & $11$ & $12$ & $13$ & $14$   \\  
		\hline 
		$p$ & $1$, $2$ & $2$ & $2$, $3$ & $3$, $4$ & $3$, $4$ & $4$, $5$ & $4$, $5$ & $5$, $6$ & $5$, $6$, $7$ & $6$, $7$ & $6$, $7$, $8$ & $7$, $8$, $9$  
		\\		
		\hline
		
	\end{tabular}
	\label{table2} 
\end{table}

\begin{proposition} \label{FormatC} If $r$ is not integer, $p< \frac{\ln 2}{\ln
		\left(\frac{r}{\lfloor r \rfloor}\right)}$, $\lfloor r \rfloor ^{p}+ \lfloor r-1 \rfloor ^{p} \leq r^{p}$ and $2 \lfloor r \rfloor^{p} \leq \lfloor r+1 \rfloor^{p}$,  then the lattice $\Lambda_r$ with basis $\{( 2 \lfloor r \rfloor +1,-1),(2\lfloor r \rfloor -1, 2 \lfloor r \rfloor)\}$ is quasi-perfect. It has discrete packing density
	$\Delta_p^{2}(\Lambda_r)=\displaystyle\frac{(2\lfloor r \rfloor +1)^{2}-4}{4\lfloor r\rfloor^2 +4 \lfloor r \rfloor -1}$. \end{proposition}
\begin{proof} The proof follows from the fact that the corners $(\pm \lfloor r \rfloor, \pm \lfloor r \rfloor)$ have $p$-norm smaller or equal to the $p$-norm of the points $(\pm \lfloor r+1 \rfloor, 0)$.  The formula for the packing density follows from Proposition \ref{pontos} since $\mu(2,p,r) =(2\lfloor r \rfloor +1)^{2}-4$. \end{proof}

%


\begin{proposition} \label{FormatD} If $r$ is not integer, $\lfloor r \rfloor^{p}+ \lfloor r-1 \rfloor^{p} >  r^{p}$, $\lfloor r \rfloor^{p}+ \lfloor r-2 \rfloor^{p} \leq  r^{p}$  and $2 \lfloor r \rfloor^{p} \leq \lfloor r+1 \rfloor^{p}$, then the lattice $\Lambda_r$ with basis $\{( 2 \lfloor r \rfloor +1,-2),(2\lfloor r \rfloor -2, 2 \lfloor r \rfloor-1)\}$ is $2$-imperfect. It has discrete packing density
	$\Delta_p^{2}(\Lambda_r)=\displaystyle\frac{(2\lfloor r \rfloor +1)^{2}-12}{4\lfloor r\rfloor^2 +4 \lfloor r \rfloor -5}$.
\end{proposition}
\begin{proof} The proof follows from the fact that the points $(\pm \lfloor r-1 \rfloor, \pm \lfloor r \rfloor)$ have $p$-norm smaller or equal to the $p$-norm of $(\pm \lfloor r \rfloor, \pm \lfloor r \rfloor)$ and these points have $p$-norm smaller or equal to the $p$-norm of the point $(\pm \lfloor r+1 \rfloor, 0)$.  The formula for the packing density follows from Proposition \ref{pontos} since $\mu(2,p,r) =(2\lfloor r \rfloor +1)^{2}-12$. \end{proof}

%

\begin{example} The lattice $\Lambda_{5.2}$ with basis  $\{(9,-1),(7,8)\}$  is $2$-imperfect in the $\ell_p$ metric for $p=4$. The lattice $\Lambda_{3.2}$ with basis  $\{(4,5), (7,-1)\}$ is associated to  the same format of the polyominoes of Proposition \ref{FormatC} but it is quasi-perfect in the $\ell_p$-metric for $p =2$. 
	
	\begin{figure}[h!]
		\begin{minipage}[b]{0.3\linewidth}
			\centering
			\includegraphics[scale=0.25]{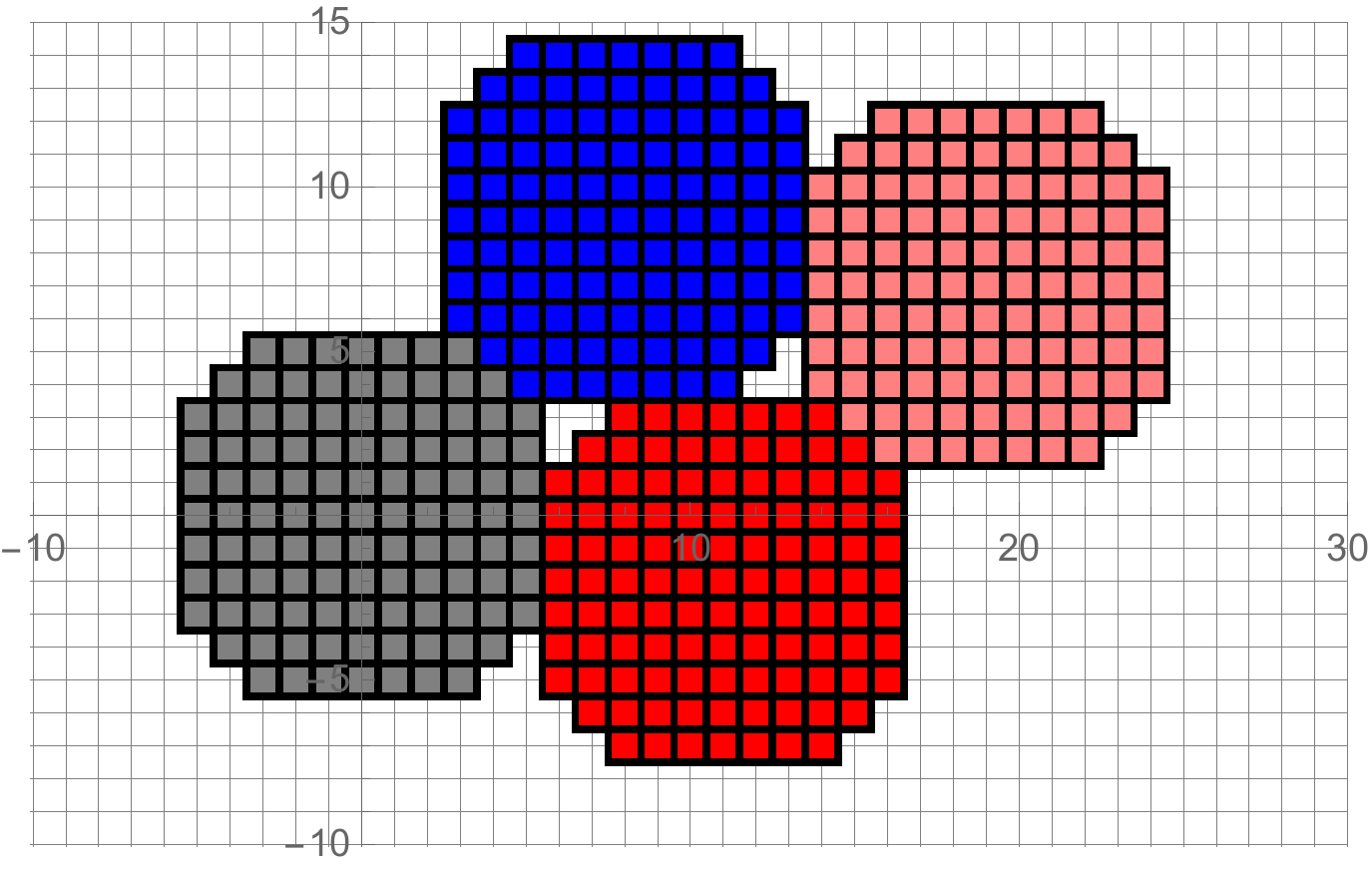}
		\end{minipage}
		\begin{minipage}[b]{0.3\linewidth}
			\centering
			\includegraphics[scale=0.25]{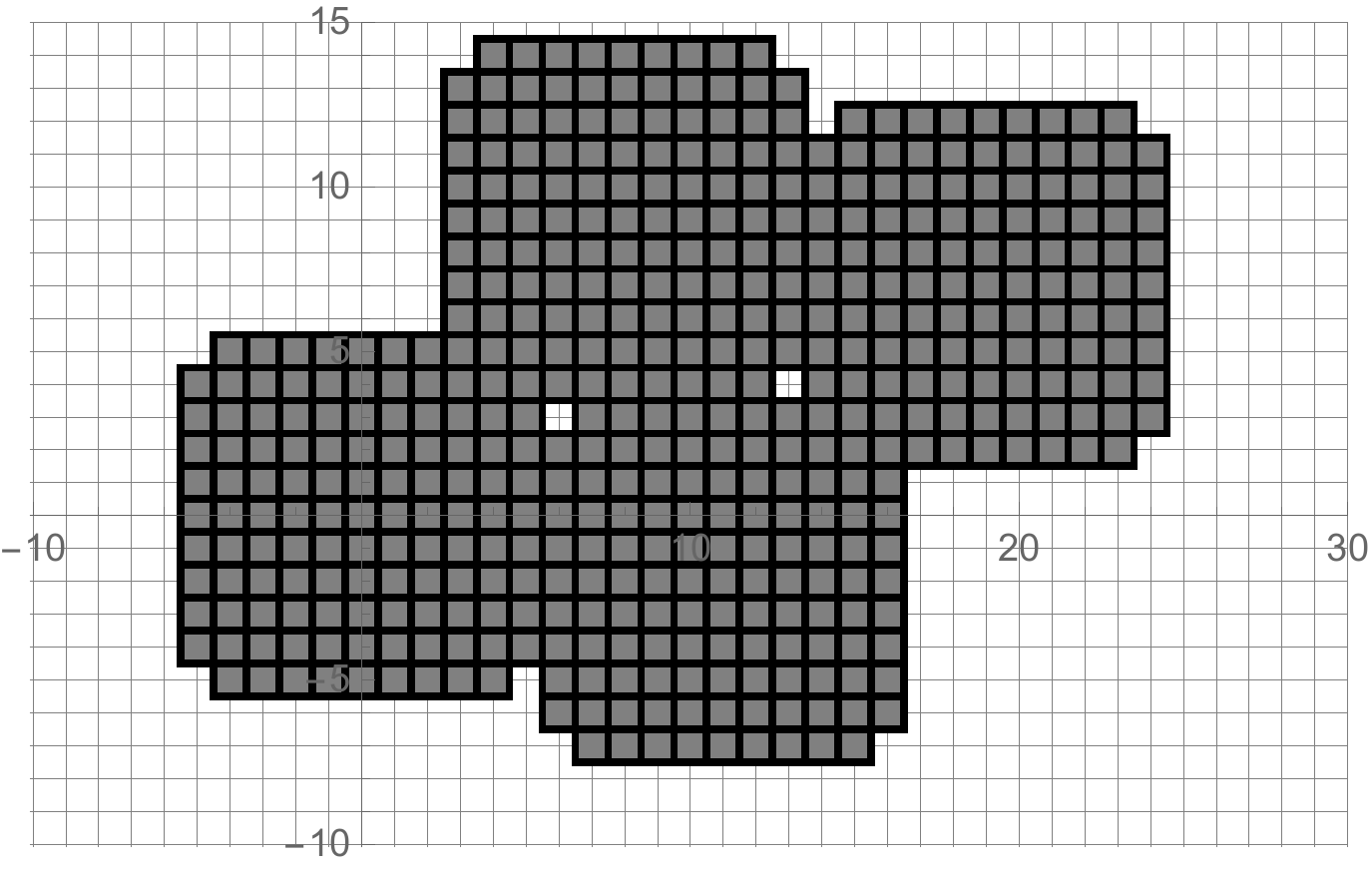}
		\end{minipage}
		\begin{minipage}[b]{0.25\linewidth}
			\centering
			\includegraphics[scale=0.21]{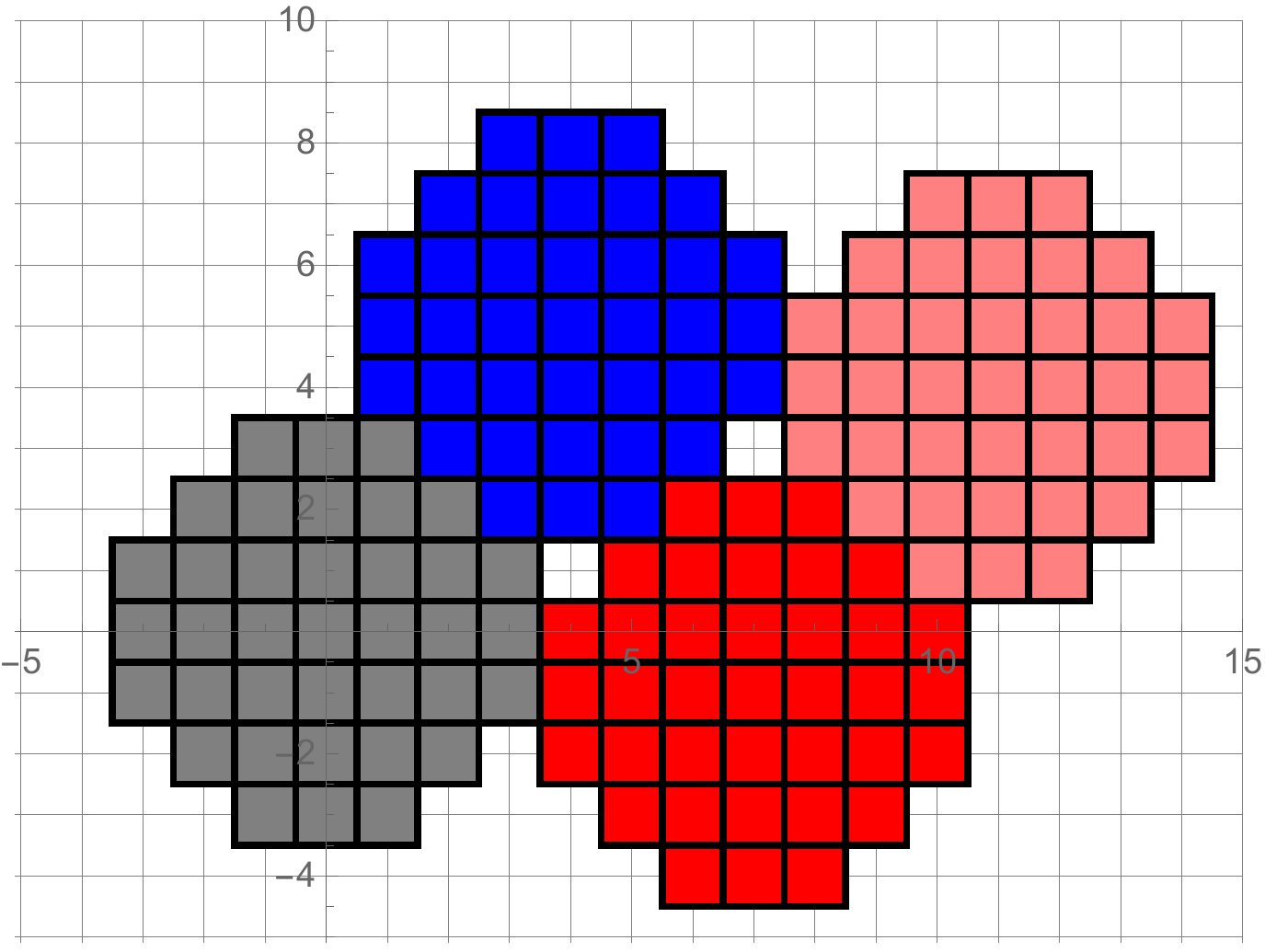}
		\end{minipage}
		\caption{From the left to the right the polyominoes associated to the balls $B_p^{2}(r)$ for $(r,p) = (5.2, 4)$, $(5.5, 4)$ and  $(3.2, 2)$, respectively.}
		\label{Fig5}
	\end{figure}
	\end{example}


The lattices $\Lambda_r$ obtained in Propositions \ref{FormatA}, \ref{FormatB}, \ref{FormatC} and \ref{FormatD} satisfy $\lim_{r\to \infty} \Delta_p^{n}(\Lambda_r) = 1$.

\section{Computational Algorithm}\label{Algorithm}

In this section we present a pseudo code of an algorithm (Algorithm \ref{algo_quasiperfect}) that lists all perfect and quasi-perfect linear codes in $\mathbb{Z}^{n}$ considering the $\ell_p$ metric, $2 \leq p < \infty$, until a certain given packing radius up to congruence. 

 In the cases that the minimum covering density is known we use Proposition \ref{bounds} to bounding our search. 

The algorithm is based on two tests where the first test is a variation of the next theorem. 

\begin{theorem}{\upshape{\cite[Thm. 6]{Diameter}}}
	Let $\mathcal{P} \subset \mathbb{Z}^n$, such that $|\mathcal{P}| = m$. There is a lattice tiling of $\mathbb{Z}^n$ by translates of $\mathcal{P}$ if and only if there is an Abelian group $G$ of order $m$ and a homomorphism $\phi: \mathbb{Z}^n \to G$ such that the restriction of $\phi$ to $\mathcal{P}$ is a bijection.
	\label{thm:Homomorphism}
\end{theorem}

The first test, called here {\it Injetivity Test}, ensures that the balls of a certain radius $r_p$ centered at points of a lattice $\Lambda$ are disjoint. More precisely, the lattice is viewed as the kernel of an application $\phi: \mathbb{Z}^n \longrightarrow \mathcal{G}$ (where $\mathcal{G}$ is an abelian group with $\#G = M = \det \Lambda$). Suppose that the lattice $\Lambda$ is generated by a matrix $A$ and consider $B$ the adjoint matrix of $A$. 
Up to group isomorphisms, we may assume that $\phi$ is the composition of two applications $\tilde{\phi}$ and $\pi$, where $\tilde\phi(x)= x B $ and $\pi(x)= \bar{x}\,\,(mod\, M)$ 
as the following diagram: $$\xymatrix{\mathbb{Z}^n \ar[r]^{\tilde\phi} \ar[rd]^{\phi} & \mathbb{Z}^n \ar[d]^{\pi }\\ & \dfrac{\mathbb{Z}^n}{M \mathbb{Z}^n}}$$ 
 If two elements of $B_{p}^{n}(r_p)$ have the same image $\phi(x)=\phi(y)$ with $x,y \in B_{p}^{n}(r_p)$, this means that the difference between them is an element of a lattice $\Lambda = Ker(\phi)$. Therefore $\phi(x)=\phi(y)$ implies $x-y = u A$, for some $u \in \mathbb{Z}^{n}$ and then $(x-y)B=u A B = u M \equiv 0\,\, (mod \, M)$, where $M = \det \Lambda$.  Summarizing, the test verify if different points of the ball $B_p^{n}(r_p)$ times $B$ are not equivalent modulus $M$.

If we have a lattice packing we must verify if this lattice is either perfect or quasi-perfect.  The second test, called here {\it Covering Test}, consists of verifying if a Voronoi cell is  a subset of the covering ball $B_p^{n}(R_p)$.

\begin{algorithm}[H]
\SetKwData{Left}{left}\SetKwData{This}{this}\SetKwData{Up}{up}
\SetKwFunction{Union}{Union}\SetKwFunction{FindCompress}{FindCompress}
\SetKwInOut{Input}{inputs}\SetKwInOut{Output}{outputs}
\Input{$M =$ volume of the lattice;  $n=$ dimension; $p \in \mathbb{N}$, $\ell_p$ metric}
\Output{List of perfect and quasi-perfect lattices of volume $M$}
\BlankLine
\Begin{
initialization\;
 $r_p \leftarrow max_{r\in \mathcal{D}_{p,n}}\{r;\,\,\# B_p^n(r)\leq M\}$\;
 $R_p \leftarrow min_{r\in \mathcal{D}_{p,n}}\{r;\,\,\# B_p^n(r)> M\}$\;
 $B_p \leftarrow B_p^n(r_p)$\;
 $B_c \leftarrow B_p^n(R_p)$\;
 $Lattices \leftarrow \{\Lambda;\,\, \Lambda \text{ sublattices of } \mathbb{Z}^n \text{ with } \mathrm{vol}(\Lambda)=M \}$\;
 $DenseLattices \leftarrow \{\}$\;
 $QuasiPerfect \leftarrow \{\}$\;
 $C \leftarrow 1$\;
 \While{$C\leq \# Lattices$}{
	\If{``Injetivity Test" in $C$-th element of $Lattices$ is positive}{
  add $C$-th element of $Lattices$ in $DenseLattices$\;}
    $C \leftarrow C+1$\;
	}
	$C \leftarrow 1$\;     
 \While{$C \leq \# DenseLattices$}{
	\If{``Covering Test" in $C$-th element of $DensitLattices$ is positive}{
  add in $C$-th element of $DenseLattices$ in $QuasiPerfect$}
    $C \leftarrow C+1$\;
	}
}
\caption{Perfect and Quasi-perfect test}
\label{algo_quasiperfect}
\end{algorithm}

\vspace{1cm}

Let ${\mathcal{A}}$ the set of generator matrices in dimension $2$ given by

\begin{tiny}
\[\begin{split} {\mathcal{A}} =  &  \left\{\left(\begin{array}{cc} 1 & 2 \\ 0 & 6 \end{array}\right)\right.,  \left(\begin{array}{cc} 1 & 2 \\ 0 & 7 \end{array}\right), \left(\begin{array}{cc} 1 & 3 \\ 0 & 7 \end{array}\right), \left(\begin{array}{cc} 1 & 3 \\ 0 & 8 \end{array}\right), \left(\begin{array}{cc} 1 & 3 \\ 0 & 11 \end{array}\right), \left(\begin{array}{cc} 1 & 4 \\ 0 & 11 \end{array}\right), \left(\begin{array}{cc} 2 & 3 \\ 0 & 6 \end{array}\right), \\& \left(\begin{array}{cc} 1 & 4 \\ 0 & 14 \end{array}\right),  \left(\begin{array}{cc} 1 & 4 \\ 0 & 15 \end{array}\right) \left(\begin{array}{cc} 1 & 6 \\ 0 & 15 \end{array}\right), \left(\begin{array}{cc} 1 & 6 \\ 0 & 16 \end{array}\right), \left(\begin{array}{cc} 1 & 4 \\ 0 & 17 \end{array}\right), \left(\begin{array}{cc} 1 & 5 \\ 0 & 17 \end{array}\right), \left(\begin{array}{cc} 1 & 7 \\ 0 & 17 \end{array}\right), \\& \left(\begin{array}{cc} 1 & 4 \\ 0 & 18 \end{array}\right), \left(\begin{array}{cc} 1 & 5 \\ 0 & 18 \end{array}\right),
	\left(\begin{array}{cc} 1 & 7 \\ 0 & 18 \end{array}\right),
	\left(\begin{array}{cc} 1 & 4 \\ 0 & 19 \end{array}\right), 
	\left(\begin{array}{cc} 1 & 5 \\ 0 & 19 \end{array}\right),
	\left(\begin{array}{cc} 1 & 8 \\ 0 & 20 \end{array}\right),
	\left(\begin{array}{cc} 1 & 5 \\ 0 & 23 \end{array}\right),\\& 
	\left(\begin{array}{cc} 1 & 9 \\ 0 & 23 \end{array}\right), 
	\left(\begin{array}{cc} 1 & 5 \\ 0 & 24 \end{array}\right),
	\left(\begin{array}{cc} 1 & 6 \\ 0 & 33 \end{array}\right), 
	\left(\begin{array}{cc} 1 & 6 \\ 0 & 34 \end{array}\right), 
	\left(\begin{array}{cc} 1 & 10 \\ 0 & 35 \end{array}\right), 
	\left(\begin{array}{cc} 1 & 7 \\ 0 & 39 \end{array}\right), 
	\left(\begin{array}{cc} 1 & 11 \\ 0 & 39 \end{array}\right), \\& \left.\left(\begin{array}{cc} 1 & 12 \\ 0 & 42 \end{array}\right)\right\} \end{split}\]
\end{tiny}	

\begin{proposition}\label{22} The quasi-perfect linear codes in $\mathbb{Z}^{2}$ in the $\ell_2$ metric, up to equivalence, have packing radius in the set $\{1,2,3,4,\sqrt{2},\sqrt{5},2 \sqrt{5},\sqrt{10}\}$ and are given by the  generator matrices listed in ${\mathcal{A}}$ and by the generator matrices \[\tiny{\begin{split}  &  \left\{\left(\begin{array}{cc} 1 & 8 \\ 0 & 53 \end{array}\right)\right., 
    \left(\begin{array}{cc} 1 & 20 \\ 0 & 53 \end{array}\right), 
    \left(\begin{array}{cc} 1 & 9 \\ 0 & 77 \end{array}\right), 
    \left.\left(\begin{array}{cc} 1 & 17 \\ 0 & 77 \end{array}\right)\right\}. \end{split}}\] \end{proposition} 
\begin{proof} Since the minimum covering radius in dimension $2$ in the $\ell_2$ metric is $1.2092$, using Inequality (\ref{7}) the maximum volume possible for a quasi-perfect linear lattice must be smaller or equal to $241$. Then, we use Algorithm \ref{algo_quasiperfect}  to list all quasi-perfect codes in dimension $2$ with volume smaller or equal to $241$.  \end{proof}

For Propositions \ref{23} and \ref{24} we use Algorithm \ref{algo_quasiperfect} for listing all quasi-perfect codes in the $\ell_p$ metric, for $p=3$ and $4$, respectively,  with volume smaller or equal to $600$.

\begin{proposition}\label{23} The linear quasi-perfect codes in $\mathbb{Z}^{2}$ in the $\ell_3$ metric with volume smaller or equal to 600, up to congruence, have packing radius in the set  $\{1,2,3,\sqrt[3]{2},3^{2/3},2^{2/3} \sqrt[3]{7},\sqrt[3]{35}\}$ and are given by generator matrices listed in ${\mathcal{A}}$ and by the generator matrices
	\[\tiny{\begin{split}  & \left\{\left(\begin{array}{cc} 1 & 7 \\ 0 & 47 \end{array}\right)\right., \left(\begin{array}{cc} 1 & 20 \\ 0 & 47 \end{array}\right), \left.\left(\begin{array}{cc} 1 & 7 \\ 0 & 48 \end{array}\right)\right\}. \end{split}}\] \end{proposition}

\begin{proposition}\label{24} The linear quasi-perfect codes in $\mathbb{Z}^{2}$ in the $\ell_4$ metric with volume smaller or equal to $600$, up to congruence, have packing radius in the set $\{1,2,3,\sqrt[4]{2},\sqrt[4]{17},\sqrt[4]{82},\sqrt[4]{97},\sqrt[4]{337}\}$ and are given by  the  generator matrices listed in ${\mathcal{A}}$ and by the generator matrices
	\[\tiny{\begin{split}  &   \left\{\left(\begin{array}{cc} 1 & 7 \\ 0 & 47 \end{array}\right)\right., \left(\begin{array}{cc} 1 & 20 \\ 0 & 47 \end{array}\right),
		\left(\begin{array}{cc} 1 & 7 \\ 0 & 48 \end{array}\right),
		\left(\begin{array}{cc} 1 & 9 \\ 0 & 79 \end{array}\right),
		\left(\begin{array}{cc} 1 & 35 \\ 0 & 79 \end{array}\right), \left.\left(\begin{array}{cc} 1 & 9 \\ 0 & 80 \end{array}\right)\right\}. \end{split}}\] \end{proposition} 

Let ${\mathcal{B}}$ the set of generator matrices in dimension $3$  given by
\[\tiny{\begin{split} {\mathcal{B}} = &  \left\{\left(\begin{array}{ccc} 1 & 0 & 2 \\ 0 & 1 & 3 \\ 0 & 0 & 8 \end{array}\right)\right.,
		\left(\begin{array}{ccc} 1 & 0 & 2  \\ 0 & 1 & 3 \\ 0 & 0 & 9 \end{array}\right),
		\left(\begin{array}{ccc} 1 & 0 & 2  \\ 0 & 1 & 4 \\ 0 & 0 & 9 \end{array}\right),
		\left(\begin{array}{ccc} 1 & 0 & 3  \\ 0 & 1 & 4 \\ 0 & 0 & 9 \end{array}\right),
		\left(\begin{array}{ccc} 1 & 1 & 1  \\ 0 & 3 & 0 \\ 0 & 0 & 3 \end{array}\right), \\&
		\left(\begin{array}{ccc} 1 & 0 & 2  \\ 0 & 1 & 3 \\ 0 & 0 & 10 \end{array}\right), 
		\left(\begin{array}{ccc} 1 & 0 & 2  \\ 0 & 1 & 4 \\ 0 & 0 & 10 \end{array}\right),
		\left(\begin{array}{ccc} 1 & 0 & 3  \\ 0 & 1 & 4 \\ 0 & 0 & 10 \end{array}\right),
		\left(\begin{array}{ccc} 1 & 0 & 2  \\ 0 & 1 & 3 \\ 0 & 0 & 11 \end{array}\right),
		\left(\begin{array}{ccc} 1 & 0 & 2  \\ 0 & 1 & 4 \\ 0 & 0 & 11 \end{array}\right), \\&
		\left(\begin{array}{ccc} 1 & 0 & 3  \\ 0 & 1 & 4 \\ 0 & 0 & 11 \end{array}\right), 
		\left(\begin{array}{ccc} 1 & 0 & 2  \\ 0 & 1 & 5 \\ 0 & 0 & 11 \end{array}\right),  
		\left(\begin{array}{ccc} 1 & 0 & 3  \\ 0 & 1 & 5 \\ 0 & 0 & 11 \end{array}\right),
		\left(\begin{array}{ccc} 1 & 0 & 4  \\ 0 & 1 & 5 \\ 0 & 0 & 11 \end{array}\right),
		\left(\begin{array}{ccc} 1 & 0 & 2  \\ 0 & 1 & 4 \\ 0 & 0 & 12 \end{array}\right), \\&
		\left(\begin{array}{ccc} 1 & 0 & 2  \\ 0 & 1 & 5 \\ 0 & 0 & 12 \end{array}\right),
		\left(\begin{array}{ccc} 1 & 0 & 3  \\ 0 & 1 & 5 \\ 0 & 0 & 12 \end{array}\right),
		\left(\begin{array}{ccc} 1 & 0 & 2  \\ 2 & 2 & 0 \\ 0 & 0 & 6 \end{array}\right),
		\left(\begin{array}{ccc} 1 & 0 & 2  \\ 0 & 2 & 3 \\ 0 & 0 & 6 \end{array}\right),
		\left(\begin{array}{ccc} 1 & 1 & 2  \\ 0 & 3 & 0 \\ 0 & 0 & 4 \end{array}\right), \\&  \left(\begin{array}{ccc} 1 & 0 & 2 \\ 0 & 1 & 4 \\ 0 & 0 & 13 \end{array}\right),
		\left(\begin{array}{ccc} 1 & 0 & 3  \\ 0 & 1 & 4 \\ 0 & 0 & 13 \end{array}\right),
		\left(\begin{array}{ccc} 1 & 0 & 2  \\ 0 & 1 & 5 \\ 0 & 0 & 13 \end{array}\right),
		\left(\begin{array}{ccc} 1 & 0 & 3  \\ 0 & 1 & 5 \\ 0 & 0 & 13 \end{array}\right),
		\left(\begin{array}{ccc} 1 & 0 & 2  \\ 0 & 1 & 6 \\ 0 & 0 & 13 \end{array}\right), \\& 
		\left(\begin{array}{ccc} 1 & 0 & 3  \\ 0 & 1 & 6 \\ 0 & 0 & 13 \end{array}\right),
		\left(\begin{array}{ccc} 1 & 0 & 4  \\ 0 & 1 & 6 \\ 0 & 0 & 13 \end{array}\right),  
		\left(\begin{array}{ccc} 1 & 0 & 2  \\ 0 & 1 & 5 \\ 0 & 0 & 14 \end{array}\right),
		\left(\begin{array}{ccc} 1 & 0 & 2  \\ 0 & 1 & 6 \\ 0 & 0 & 14 \end{array}\right),
		\left(\begin{array}{ccc} 1 & 0 & 3  \\ 0 & 1 & 6 \\ 0 & 0 & 14 \end{array}\right),\\&
		\left(\begin{array}{ccc} 1 & 0 & 4  \\ 0 & 1 & 6 \\ 0 & 0 & 14 \end{array}\right),
		\left(\begin{array}{ccc} 1 & 0 & 2  \\ 0 & 1 & 5 \\ 0 & 0 & 15 \end{array}\right),
		\left(\begin{array}{ccc} 1 & 0 & 3  \\ 0 & 1 & 5 \\ 0 & 0 & 15 \end{array}\right),
		\left(\begin{array}{ccc} 1 & 0 & 2  \\ 0 & 1 & 6 \\ 0 & 0 & 15 \end{array}\right),
		\left(\begin{array}{ccc} 1 & 0 & 3  \\ 0 & 1 & 6 \\ 0 & 0 & 15 \end{array}\right), \\&
		\left(\begin{array}{ccc} 1 & 0 & 4  \\ 0 & 1 & 6 \\ 0 & 0 & 15 \end{array}\right),
		\left(\begin{array}{ccc} 1 & 0 & 3  \\ 0 & 1 & 7 \\ 0 & 0 & 15 \end{array}\right),
		\left(\begin{array}{ccc} 1 & 0 & 5  \\ 0 & 1 & 7 \\ 0 & 0 & 15 \end{array}\right),
		\left(\begin{array}{ccc} 1 & 0 & 2  \\ 0 & 3 & 0 \\ 0 & 0 & 5 \end{array}\right),
		\left(\begin{array}{ccc} 1 & 1 & 2  \\ 0 & 3 & 0 \\ 0 & 0 & 15 \end{array}\right), \\&
		\left(\begin{array}{ccc} 1 & 0 & 2  \\ 0 & 1 & 6 \\ 0 & 0 & 16 \end{array}\right),
		\left(\begin{array}{ccc} 1 & 0 & 2  \\ 0 & 1 & 6 \\ 0 & 0 & 17 \end{array}\right),
		\left(\begin{array}{ccc} 1 & 0 & 3  \\ 0 & 1 & 6 \\ 0 & 0 & 17 \end{array}\right),
		\left(\begin{array}{ccc} 1 & 0 & 3  \\ 0 & 1 & 8 \\ 0 & 0 & 17 \end{array}\right),
		\left(\begin{array}{ccc} 1 & 0 & 3  \\ 0 & 1 & 8 \\ 0 & 0 & 21 \end{array}\right), \\&
		\left(\begin{array}{ccc} 1 & 0 & 3  \\ 0 & 1 & 8 \\ 0 & 0 & 23 \end{array}\right),
		\left(\begin{array}{ccc} 1 & 0 & 5  \\ 0 & 1 & 8 \\ 0 & 0 & 23 \end{array}\right),
		\left(\begin{array}{ccc} 1 & 0 & 3  \\ 0 & 1 & 9 \\ 0 & 0 & 23 \end{array}\right),
		\left(\begin{array}{ccc} 1 & 0 & 3  \\ 0 & 1 & 8 \\ 0 & 0 & 24 \end{array}\right),
		\left(\begin{array}{ccc} 1 & 0 & 5  \\ 0 & 1 & 8 \\ 0 & 0 & 24 \end{array}\right),  \end{split}}\] 
	\[\tiny{\begin{split}  & 	\left(\begin{array}{ccc} 1 & 1 & 3  \\ 0 & 3 & 0 \\ 0 & 0 & 8 \end{array}\right),
		\left(\begin{array}{ccc} 1 & 0 & 5  \\ 0 & 1 & 8 \\ 0 & 0 & 25 \end{array}\right),
		\left(\begin{array}{ccc} 1 & 0 & 5  \\ 0 & 1 & 9 \\ 0 & 0 & 25 \end{array}\right),
		\left(\begin{array}{ccc} 1 & 0 & 8  \\ 0 & 1 & 11 \\ 0 & 0 & 25 \end{array}\right),
		\left(\begin{array}{ccc} 1 & 0 & 5  \\ 0 & 1 & 9 \\ 0 & 0 & 26 \end{array}\right), \\& \left(\begin{array}{ccc} 1 & 0 & 11  \\ 0 & 1 & 16 \\ 0 & 0 & 35 \end{array}\right),
		\left(\begin{array}{ccc} 1 & 0 & 12  \\ 0 & 1 & 18 \\ 0 & 0 & 39 \end{array}\right),
		\left(\begin{array}{ccc} 1 & 0 & 5  \\ 0 & 1 & 13 \\ 0 & 0 & 41 \end{array}\right),
		\left(\begin{array}{ccc} 1 & 0 & 8  \\ 0 & 1 & 19 \\ 0 & 0 & 41 \end{array}\right),
		\left(\begin{array}{ccc} 1 & 0 & 13  \\ 0 & 1 & 19 \\ 0 & 0 & 41 \end{array}\right), \\&
		\left(\begin{array}{ccc} 1 & 0 & 9  \\ 0 & 2 & 6 \\ 0 & 0 & 21 \end{array}\right),
		\left(\begin{array}{ccc} 1 & 0 & 4  \\ 0 & 3 & 7 \\ 0 & 0 & 14 \end{array}\right),
		\left(\begin{array}{ccc} 1 & 3 & 2  \\ 0 & 6 & 0 \\ 0 & 0 & 7 \end{array}\right),
		\left(\begin{array}{ccc} 1 & 0 & 14  \\ 0 & 1 & 20 \\ 0 & 0 & 44 \end{array}\right),
		\left(\begin{array}{ccc} 1 & 0 & 7  \\ 0 & 1 & 11 \\ 0 & 0 & 45 \end{array}\right), \\&
		\left(\begin{array}{ccc} 1 & 0 & 8  \\ 0 & 1 & 13 \\ 0 & 0 & 45 \end{array}\right),
		\left(\begin{array}{ccc} 1 & 0 & 4  \\ 0 & 1 & 17 \\ 0 & 0 & 45 \end{array}\right),
		\left(\begin{array}{ccc} 1 & 0 & 8  \\ 0 & 1 & 12 \\ 0 & 0 & 46 \end{array}\right),
		\left(\begin{array}{ccc} 1 & 0 & 5  \\ 0 & 1 & 13 \\ 0 & 0 & 47 \end{array}\right),
		\left(\begin{array}{ccc} 1 & 0 & 4  \\ 0 & 1 & 18 \\ 0 & 0 & 47 \end{array}\right), \\&
		\left(\begin{array}{ccc} 1 & 0 & 12  \\ 0 & 1 & 19 \\ 0 & 0 & 47 \end{array}\right),
		\left(\begin{array}{ccc} 1 & 0 & 8  \\ 0 & 1 & 12 \\ 0 & 0 & 50 \end{array}\right),
		\left(\begin{array}{ccc} 1 & 0 & 5  \\ 0 & 1 & 21 \\ 0 & 0 & 55 \end{array}\right),
		\left(\begin{array}{ccc} 1 & 0 & 5  \\ 0 & 1 & 25 \\ 0 & 0 & 63 \end{array}\right),
		\left(\begin{array}{ccc} 1 & 0 & 5  \\ 0 & 1 & 26 \\ 0 & 0 & 65 \end{array}\right), \\&
		\left.\left(\begin{array}{ccc} 1 & 1 & 5  \\ 0 & 5 & 0 \\ 0 & 0 & 13 \end{array}\right)\right\}. \end{split}}\]

\begin{proposition}\label{32}  The quasi-perfect linear codes in $\mathbb{Z}^{3}$ in the $\ell_2$ metric, up to congruence, have packing radius in the set $\{1,2,\sqrt{2},2 \sqrt{2},\sqrt{5}\}$ and are given by the  generator matrices listed in ${\mathcal{B}}$ and by the generator matrices
\[\tiny{\begin{split}  &  \left\{\left(\begin{array}{ccc} 1 & 0 & 5  \\ 0 & 1 & 41 \\ 0 & 0 & 105 \end{array}\right)\right., 
\left.\left(\begin{array}{ccc} 1 & 1  & 5 \\ 0 & 6 & 6 \\ 0 & 0  & 18 \end{array}\right)\right\}. \end{split}}\]
\end{proposition} 
\begin{proof} Since the minimum covering radius in dimension $3$ in the $\ell_2$ metric is $1.4635$, using Inequality (\ref{7}) the maximum volume possible for a quasi-perfect linear lattice must be smaller or equal to $1419$. Then, we use Algorithm \ref{algo_quasiperfect}  to list all quasi-perfect codes in dimension $2$ with volume smaller or equal to $1419$.  \end{proof}

For Propositions \ref{33} and \ref{34} we use Algorithm \ref{algo_quasiperfect} for listing all quasi-perfect codes in the $\ell_p$ metric, for $p=3$ and $4$, respectively,  with volume smaller or equal to $1500$.

\begin{proposition}\label{33}  The linear quasi-perfect codes in $\mathbb{Z}^{3}$ in the $\ell_3$ metric with volume smaller or equal to $1500$, up to congruence, have packing radius in the set $\{1,2,\sqrt[3]{2},2 \sqrt[3]{2},3^{2/3},\sqrt[3]{17}\}$ and are given by the  generator matrices listed in ${\mathcal{B}}$ and by the generator matrices
	\[\tiny{\begin{split}  &  \left\{\left(\begin{array}{ccc} 1 & 1 & 5  \\ 0 & 6 & 6 \\ 0 & 0 & 18 \end{array}\right)\right., \left(\begin{array}{ccc} 1 & 0 & 5  \\ 0 & 1 & 25 \\ 0 & 0 & 123 \end{array}\right), \left(\begin{array}{ccc} 1 & 0 & 5  \\ 0 & 1 & 49 \\ 0 & 0 & 123 \end{array}\right), \left(\begin{array}{ccc} 1 & 0 & 49  \\ 0 & 1 & 59 \\ 0 & 0 & 1-23 \end{array}\right), \left.\left(\begin{array}{ccc} 1 & 0 & 5  \\ 0 & 1 & 25 \\ 0 & 0 & 124 \end{array}\right)\right\}. \end{split}}\]
	\end{proposition}

\begin{proposition}\label{34}  The linear quasi-perfect codes in $\mathbb{Z}^{3}$ in the $\ell_4$ metric with volume smaller or equal to $1500$, up to congruence, have packing radius in the set $\{1,2,\sqrt[4]{2},2 \sqrt[4]{2},\sqrt[4]{17},\sqrt[4]{33},\sqrt[4]{178}\}$ and are given by the  generator matrices listed in ${\mathcal{B}}$ and by the generator matrices
		\[\tiny{\begin{split}  &  \left\{\left(\begin{array}{ccc} 1 & 1 & 5  \\ 0 & 6 & 6 \\ 0 & 0 & 18 \end{array}\right)\right., \left(\begin{array}{ccc} 1 & 0 & 5  \\ 0 & 1 & 25 \\ 0 & 0 & 123 \end{array}\right), \left(\begin{array}{ccc} 1 & 0 & 5  \\ 0 & 1 & 49 \\ 0 & 0 & 123 \end{array}\right), \left(\begin{array}{ccc} 1 & 0 & 49  \\ 0 & 1 & 59 \\ 0 & 0 & 123 \end{array}\right), \left(\begin{array}{ccc} 1 & 0 & 5  \\ 0 & 1 & 25 \\ 0 & 0 & 124 \end{array}\right),  \\&
		\left(\begin{array}{ccc} 1 & 0  & 7 \\ 0 & 1 & 49 \\ 0 & 0  & 341 \end{array}\right), \left(\begin{array}{ccc} 1 & 0  & 7 \\ 0 & 1 & 146 \\ 0 & 0  & 341 \end{array}\right), \left(\begin{array}{ccc} 1 & 0  & 346 \\ 0 & 1 & 167 \\ 0 & 0  & 341 \end{array}\right), 	\left.\left(\begin{array}{ccc} 1 & 0  & 7 \\ 0 & 1 & 49 \\ 0 & 0  & 342 \end{array}\right)\right\}.  \end{split}}\]
	\end{proposition} 

We conjecture that the lattices listed in Propositions \ref{23}, \ref{24}, \ref{33} and \ref{34} are the only quasi-perfect codes, up to congruence, for  $n=2, 3$ and $p = 3,4$. 

\bibliography{biblio}
\bibliographystyle{plain}

\end{document}